\documentclass[12pt]{article}
\usepackage{geometry, booktabs, subcaption}
\usepackage{titlesec}
\usepackage[T1]{fontenc}
\usepackage{authblk}
\usepackage{amssymb,amsmath, amsthm}
\usepackage{graphicx}
\usepackage{multirow}
\usepackage{amsmath}
\usepackage{graphicx}
\usepackage{natbib}
\usepackage{url} 

\usepackage{amsmath,amsfonts,amsthm,latexsym,graphicx,mathrsfs,xcolor,amssymb}
\usepackage{mathtools}
\usepackage{bm}
\usepackage{cleveref}
\usepackage{subcaption}
\usepackage{placeins}

\newtheorem{proposition}{Proposition}

\newcommand{\calY}{\mathcal Y}
\newcommand{\calZ}{\mathcal Z}

\newcommand{\E}{\mathbb E}

\newcommand{\R}{\mathbb R}

\newcommand{\Y}{\mathbb Y}
\newcommand{\Z}{\mathbb Z}

\newcommand{\prob}{\mathrm{Pr}}

\newcommand{\dd}{\mathrm{d}}
\newcommand{\indicator}{\mathrm{I}}

\newcommand{\iid}{\stackrel{\mbox{\scriptsize iid}}{\sim}}
\newcommand{\ind}{\stackrel{\mbox{\scriptsize ind}}{\sim}}

\renewcommand{\mid}{\,|\,}
\renewcommand{\vspace}[1]{}

\usepackage[normalem]{ulem}
\newcommand{\MBtext}[1]{{\color{black}{#1}}} 

\usepackage[linesnumbered,ruled,vlined]{algorithm2e}
\makeatletter
\renewcommand{\algocf@captiontext}[2]{#1\algocf@typo. \AlCapFnt{}#2} 
\def\@algocf@capt@plain{top}
\renewcommand{\algocf@makecaption}[2]{%
	\addtolength{\hsize}{\algomargin}%
	\sbox\@tempboxa{\algocf@captiontext{#1}{#2}}%
	\ifdim\wd\@tempboxa >\hsize
	\hskip .5\algomargin%
	\parbox[t]{\hsize}{\algocf@captiontext{#1}{#2}}
	\else%
	\global\@minipagefalse%
	\hbox to\hsize{\box\@tempboxa}
	\fi%
	\addtolength{\hsize}{-\algomargin}%
}
\makeatother

\def\spacingset#1{\renewcommand{\baselinestretch}%
{#1}\small\normalsize} \spacingset{1}

\allowdisplaybreaks

\title{MCMC for Bayesian nonparametric mixture modeling under differential privacy}

\author[1]{Mario Beraha}
\affil[1]{Department of Mathematics, Politecnico di Milano}

\author[1,2]{Stefano Favaro}
\affil[2]{Collegio Carlo Alberto}

\author[3]{Vinayak Rao}
\affil[3]{Department of Statistics, Purdue University}

\begin{document}

\maketitle

\begin{abstract}
Estimating the probability density of a population while preserving the privacy of individuals in that population is an important and challenging problem that has received considerable attention in recent years. While the previous literature focused on frequentist approaches, in this paper, we propose a Bayesian nonparametric mixture model under differential privacy (DP) and present two Markov chain Monte Carlo (MCMC) algorithms for posterior inference.
One is a marginal approach, resembling Neal's algorithm 5 with a pseudo-marginal Metropolis-Hastings move, and the other is a conditional approach.
Although our focus is primarily on local DP, we show that our MCMC algorithms can be easily extended to deal with global differential privacy mechanisms.
Moreover, for some carefully chosen mechanisms and mixture kernels, we show how auxiliary parameters can be analytically marginalized, allowing standard MCMC algorithms (i.e., non-privatized, such as Neal's Algorithm 2) to be efficiently employed.
Our approach is general and applicable to any mixture model and privacy mechanism.
In several simulations and a real case study, we discuss the performance of our algorithms and evaluate different privacy mechanisms proposed in the frequentist literature. 
\end{abstract}

\noindent%
{\it Keywords:} Dirichlet Process, Data augmentation, Pseudo-marginal MCMC


\section{Introduction}

Mixture models offer a natural and computationally convenient framework for density estimation and clustering \citep{Fru19}. Bayesian nonparametric (BNP) mixture models, introduced in the seminal work of \cite{Lo84}, assume that each observation belongs to one of a potentially infinite number of groups or components, each modeled with parametric density function $f(\cdot \mid \phi)$ for some parameter $\phi$. Formally, $n\geq1$ observations are modeled as a random sample $Y=(Y_1, \ldots, Y_n)$ distributed as
$    Y_i \mid P \iid \int_{\Theta} f(\cdot \mid \theta) P(\dd \theta)$,
with the mixing distribution $P$ being an almost surely discrete random probability measure, i.e. 
 $   P(\cdot) = \sum_{h \geq 1} w_{h} \delta_{\phi_h}(\cdot)$,
with the $w_{h}$'s providing the relative prevalences of the mixture components labelled by the $\phi_{h}$'s. Mixture models are increasingly popular in applied contexts such as sociology \citep{land2001introduction, collins2009latent}, healthcare analysis \citep{singh2010inpatient}, multi-omics data \citep{wang2020brem}, and many other contexts involving sensitive data \citep{schlattmann2009medical}. 
\MBtext{In such contexts, mixtures can be used to identify homogeneous subpopulations in the data, i.e., clusters, and leverage such information to develop personalized medical treatments \citep{pedone2024personalized}, or used to model nuisance parameters in a semiparametric framework.}
In all these settings, there is a need to balance learning the mixing distribution $P$ with protecting individual's privacy. 

In this paper, we consider the statistical problem of fitting a nonparametric Dirichlet process mixture model when  one has access only to a differentially-private (DP)~\citep{dwork2006differential} database $Z$, obtained from the original confidential database $Y$ via some publicly known privatization mechanism.
We contrast our problem with so-called algorithmic approaches~\citep{dwork_roth}, whose goal is to propose randomized algorithms that preserve individual privacy. Algorithmic approaches for (Gaussian) mixture models under DP have been proposed starting from \cite{nissim2007smooth}; see also \cite{kamath2019differentially} and the references therein. 
In the Bayesian framework, \cite{dimitrakakis2017differential}, \cite{savitsky2022bayesian}, and \cite{hu2022mechanisms} studied the privacy properties of (pseudo) posteriors and established 
conditions under which samples from the posterior or from the posterior predictive distribution yields a perturbed database $Z$ with valid privacy guarantees.
\cite{bernstein2018differentially, bernstein2019differentially} focussed on likelihoods belonging to exponential family and linear regression, and showed that by perturbing sufficient statistics appropriately, 
and using a suitable asymptotic approximation, the posterior distribution given $Z$ is easy to compute.

When it comes to estimation given privatized data, in the frequentist literature, the focus is usually on estimators for quantities of interest, or statistical tests based on $Z$, while also quantifying the loss of efficiency due to privacy in a minimax framework 
\citep[see][and the references therein]{wasserman2010statistical, duchi_minimax_2018}. Then, one usually seeks an optimal pair of estimator and perturbation mechanism that match the minimax rate.
\MBtext{By contrast, we assume that the mechanism is given and propose two algorithms for posterior inference that are not specific to the mechanisms.
The advantages of taking a Bayesian approach are multi-faceted. 
First, it makes it straightforward to bring prior information in the model.
Second, Bayesian estimators are naturally endowed with uncertainty quantification via the posterior distribution. Third, it is straightforward to embed a model in more complex, often semiparametric, ones, e.g., by assuming a mixture likelihood for the error distribution in a regression model, and taking into account additional information coming from  exogenous covariates by assuming, e.g., one of the variations of the dependent Dirichlet process prior \citep{quintana2022dependent}.}
Close to our work, \cite{karwa2015private} and \cite{ju_data_2022} treat the $Y_i$'s as missing variables and develop general variational Bayes and MCMC algorithms for posterior inference. One of our proposed samplers, the conditional sampler, is an instance of the algorithm of~\citet{ju_data_2022}, though our other sampler, the marginal sampler, takes a pseudo-marginal MCMC approach, allowing us to better trade off computation and mixing by introducing additional auxiliary variables.

\subsection{Our contributions}
In this paper, we investigate posterior inference in BNP mixture models under differential privacy. We start by focusing on the notion of \textit{local DP} \citep{dwork2006differential}, where the sanitized $Z = (Z_1, \ldots, Z_n)$ is obtained by perturbing each $Y_i$ individually.
Here, if the original dataset $Y$ is modeled as a realization of a BNP mixture model, then, after marginalizing out the $Y_i$'s, the $Z_i$'s again follow a nonparametric mixture model, whose kernel in general intractable.
We propose two MCMC sampling schemes for posterior density estimation and clustering given $Z$. The first is a  \emph{marginal} algorithm, where we marginalize out the infinite-dimensional measure $P$, and  can be considered a version of Algorithm 5 in \cite{neal2000markov} for private data. Here, our approach has the flavor of the pseudo-marginal MCMC approach \citep{And09}, 
with the latent $Y_i$'s introduced back into the MCMC state as auxiliary variables to deal with intractable Metropolis-Hastings probabilities. The second algorithm is a \emph{conditional} algorithm, and it can be considered as a private version of the slice-sampler of \cite{kalli2011slice}. In this setting, the $Y_i$'s are part of the state of the MCMC and we adopt the same blocked-Gibbs sampling strategy as \cite{ju_data_2022}, alternately sampling $P$ given the $Y_i$'s and the $Y_i$'s given $P$ and the $Z_i$'s. For specific local DP channels (including the Gaussian mechanism), we show how one can 
combine the privacy mechanism with a carefully chosen mixture kernel and base measure, allowing us to analytically marginalize with respect to the $Y_i$'s and obtain a conjugate mixture model of the $Z_i$'s. Then, the posterior distribution can be sampled using efficient algorithms such as Algorithm 3 in \cite{neal2000markov} or the split-merge algorithm in \cite{jain2004split}. We also extend our framework to the setting of \emph{global DP}, where the sanitized database $W = (W_1, \ldots, W_k)$ is obtained by perturbing the sample $Y = (Y_1, \ldots, Y_n)$ jointly, and show that the two algorithms proposed can be adapted in a straightforward way to this setting.

For concreteness, we consider BNP mixture models based of a Dirichlet process prior for $P$ \citep{Fer73, Lo84}, such that the random probabilities $w_{h}$'s follows the stick-breaking construction \citep{Set94}. However, the ideas developed in this paper can be easily extended to more general prior distributions, under which marginal and conditional MCMC sampling schemes for BNP mixture models have been developed in the non-private setting \citep{Griffin11,Bar13,Fav13,Fav12,Lom17,miller2018mixture,ArDeInf19}. See \cite{wade_review} for a recent review.
Our MCMC sampling schemes do not depend on a specific choice of the DP channel used to produce the sanitized database $Z$, and we consider popular channels such as adding Laplace and Gaussian noise, wavelet-based perturbation \citep{duchi_minimax_2018, butucea2020local}, and sampling from a perturbed histogram \citep{wasserman2010statistical}. 


\MBtext{We find that except when the privacy requirement is low, the marginal algorithm produces better mixing MCMC chains. However, for large sample sizes, the conditional algorithm is more than 10 times faster, and this can potentially compensate for the slower mixing. This is similar to what happens when directly working with the confidential data.}


The paper is structured as follows. In Section \ref{sec:privacy} we present an overview of DP, focussing on some definitions of local DP and DP channels. Section \ref{sec:model} contains the main results of the paper, introducing the private Neal 5 algorithm, the private slice-sampling and some variations thereof. In Section \ref{sec:num} we present a simulation study comparing the proposed algorithms over different choices of the DP channel, and Section \ref{sec:blood_donors} contains an application to real data on blood donors at the Milano Department of the Associazione Volontari Italiani del Sangue (AVIS). In Section \ref{sec:disc} we discuss the limitations of the proposed approach, and present some further research directions. Our algorithms have been implemented in \texttt{C++} as part of the \texttt{BayesMix} library \citep{beraha2022bayesmix} and interested users can easily extend our examples to different prior distributions, mixture kernels and/or privacy mechanisms.


\section{Local Differential Privacy}\label{sec:privacy}

We present a brief overview of DP, focusing on  \emph{local} and \emph{non-interactive} DP, arguably the most popular notion of privacy in the statistical framework of density estimation \citep{butucea2020local, kroll21, butucea_interactive_2022, sart2023}. For any $n\geq1$, denote by $Y_1, \ldots, Y_n$ the confidential observations, defined on a probability space $(\Omega, \mathcal A, \prob)$, and taking values in a general measurable (Polish) space $(\Y, \calY)$. Under local DP,  the sample $Y = (Y_1, \ldots, Y_n)$ is randomly perturbed to produce a \emph{sanitized} database $Z = (Z_1, \ldots, Z_n) \in \mathbb Z^n$, with the perturbation being encoded by a collection of \emph{channels}  $Q_{1},\ldots,Q_{n}$, where $Q_i: \mathbb Z \times \mathbb Y \rightarrow \R_+$ is such that,  for any $y \in \mathbb Y$, $Q_i(\cdot \mid y) := Q_i(\cdot , y)$ is a conditional distribution with density function  over $\Z$, $i=1,\ldots,n$. The \emph{sanitized} database $Z \sim \otimes_{i=1}^n Q_i(\cdot \mid Y_i)$ is released for public use. Note that local DP does not require a trusted data holder even for collecting the data, since each individual can perturb their datum $Y_i$ independently of all the other individuals.  Depending on the definition of local DP, the channels $Q_i$'s need to satisfy some properties. We recall the most popular definitions below.

In her seminal work, \cite{dwork2006differential} introduced the framework of $\varepsilon$-DP, where
\begin{equation}\label{eq:dp_def}
\sup_{S \in \calZ} \frac{Q_i(S \mid Y_i = y)}{Q_i(S \mid Y_i = y^\prime)} \leq e^\varepsilon \quad \text{for any $y, y^\prime \in \Y$}.
\end{equation}
The parameter $\epsilon$ controls the \emph{privacy loss} that an individual can suffer when analyzing their data. See \cite{bun2016concentrated} for further discussions. As shown in \cite{wasserman2010statistical}, $\varepsilon$-DP also implies an upper bound on the power of the statistical test aimed at identifying the presence of a particular individual in the unobserved original sample. Among several relaxations of $\varepsilon$-DP, the most popular is $(\varepsilon, \delta)$-DP, requiring:
\begin{equation}\label{eq:eps_delta_dp}
    Q_i(S \mid Y_i = y) \leq e^{\varepsilon} Q_i(S \mid Y_i = y^\prime) + \delta, \quad \text{for all measurable  $S$, and for any $y, y^\prime \in \Y$.} 
\end{equation}
According to this definition, DP holds \emph{most of the times},  with the probability that individuals suffer a privacy loss greater than $\varepsilon$ bounded by $\delta$. Another relaxation, referred to as concentrated DP \citep{dwork2016concentrated, bun2016concentrated}, has gained significant interest, to the extent of being adopted by the US Census Bureau  \citep{Garfinkel2022Differential}. We say that a channel $Q_i$ is $\rho$ zero concentrated DP ($\rho$-zCDP) if  
$
    D_\alpha\left(Q_i(\cdot \mid y) || Q_i(\cdot \mid y^\prime)\right) \leq  \rho \alpha \quad \text{for any $y, y^\prime \in \Y$, }
$
where $D_\alpha$ is the $\alpha$-R\'enyi divergence. We refer to \citet{dwork2016concentrated} and \citet{bun2016concentrated} 
for additional details. 

\subsection{Local DP: channels}\label{sec:channels}
We present some popular choices of the local DP perturbation channel $Q$ we will consider. The most natural one adds independent and identically distributed zero-mean noise:\begin{displaymath}
    Z_i = Y_i + \gamma_i, \qquad \gamma_i \iid f_{\gamma}.
\end{displaymath}
To match any of the notions of DP discussed above, we must require that the $Y_i$'s have a bounded support. In particular, consider unidimensional data on $\mathbb{Y}$, and denote by $\Delta = \sup_{y, y^\prime} |y - y^\prime|$ the diameter of $\mathbb Y$. If $\gamma_i \sim \mathcal{L}(0,\Delta/\varepsilon)$, with $\mathcal{L}$ the Laplace distribution, then $Q_i$ satisfies $\varepsilon$-DP, and hence $(\varepsilon, \delta)$-DP for any choice of $\delta$, as well as $\frac{1}{2} \varepsilon^2$-zCDP \citep[][Propositon 1.4]{bun2016concentrated}. The use of the Gaussian distribution is also popular, that is $\gamma_i \sim \mathcal N(0, \eta^2)$, with $\mathcal{N}$ being the Gaussian distribution. The choice of the Gaussian distribution is known to satisfy $(\varepsilon, \delta)$-DP \citep{dwork_roth, balle_wang}, for suitable values of $\varepsilon$ and $\delta$, and $\rho$-zCDP, though it does not satisfy $\varepsilon$-DP. 

Other DP mechanisms, not based on adding random noise, have been proposed in the literature. 
The minimax framework of \cite{duchi_minimax_2018} seeks the optimal pair of perturbation and estimator for a given problem, thus making the optimal estimator and perturbation mechanism related \citep{duchi_minimax_2018,butucea2020local,butucea_interactive_2022,kroll21,sart2023}. 
We review the wavelet-based perturbation mechanism of \cite{butucea2020local} in greater detail in Appendix \ref{app:background}.
Later, we will focus on one particular example of the this based on the Haar basis, whereby $Z_i = (Z_{i, j, k})$ for $j=0, \ldots, J$ and $k \leq 2^j$ is distributed as
\begin{equation}\label{eq:wav_priv}
    Z_{i, j, k} \mid Y_i \sim \mathcal{L}(\psi_{j, k}(Y_i), s),
\end{equation}
where $\psi_{j, k}(x) = 2^{j/2}\left(  \indicator_{[0, 1/2)}(2^j x - k)  - \indicator_{[1/2, 1)}(2^j x - k) \right)$, $s = \frac{12}{\varepsilon} \frac{\sqrt{2}}{\sqrt{2} - 1} 2^{J/2}$.
Here $J$ is a parameter that trades-off the amount of information released and the amount of noise added.
This perturbation guarantees $\varepsilon$-DP \citep[][Equation 3.1]{butucea2020local}.


\section{BNP mixture modeling under local DP}\label{sec:model}

We assume a BNP mixture model for the confidential $Y_i$'s.  Taking into account the local DP channels $Q_i$'s, we write the following hierarchical model for sanitized observations:
\begin{equation}\label{eq:bnp_mix}
\begin{aligned}
    Z_i \mid Y_i & \ind Q_i(\cdot \mid Y_i), \qquad &i=1, \ldots, n \\
     Y_i \mid \theta_i & \ind f(\cdot \mid \theta_i), \ \ \theta_i \mid P  \iid P,  \qquad &i=1, \ldots, n \\
    P & \sim \mathscr Q. &
\end{aligned}
\end{equation}
Here, $f(\cdot \mid \theta)$ is a probability density kernel indexed by a parameter $\theta \in \Theta$, and $P$ is an almost surely discrete random probability measure over $\Theta$, with its law $\mathscr Q$ being the Dirichlet process with concentration parameter $\alpha$ and base measure $G_0$ \citep{Fer73}. 
Thus, $P = \sum_{h \geq 1} w_h \delta_{\phi_h}$ where $(\phi_h)_{h \geq 1}$ is a collection of i.i.d. random variables from $G_0$ and $(w_h)_{h \geq 1}$ is a sequence of positive random weights summing to one, obtained, for instance, via the so-called stick breaking process:  $w_h = \nu_h \prod_{j < h} (1 - \nu_h)$, where $\nu_h \iid \mbox{Beta}(1, \alpha)$.

\MBtext{MCMC computations for mixture models when dealing directly with the $Y_i$'s are well-established. The cornerstone paper \cite{neal2000markov} proposed 8 different algorithms that fall in the class of \emph{marginal} samplers, whereby the mixing measure $P$ is integrated from the model. Another class of algorithms are the so-called \emph{conditional} ones, exemplified by the slice sampler by \cite{kalli2011slice} and the blocked-Gibbs sampler by \cite{ishwaran2001gibbs}, which rely on a (possibly random) truncation of the mixing measure $P$.
Generally speaking, marginal samplers tend to produce better-mixing chains since they operate on fewer parameters but are inherently sequential. Conditional ones, instead, can be parallelized and tend to have faster runtimes but exhibit slower mixing. See, e.g., \cite{beraha2022bayesmix} for a detailed performance comparison.}

{\spacingset{1.0}
\begin{table}[t!]
\centering
 \begin{tabular}{c || c | c | c} 
 \hline
 Algorithm & Data & State & Requirements \\ [0.5ex] 
 \hline
 \multirow{2}{5em}{Neal 2} & Original  & $c_i$'s, $\theta^*_h$'s & $f$ conjugate to $G_0$ \\
 & Sanitized & $c_i$'s, $\theta^*_h$'s & $g_i$ tractable and conjugate to $G_0$  \\
 \hline
 \multirow{2}{5em}{Neal 3} & Original  & $c_i$'s & $f$ conjugate to $G_0$ \\
 & Sanitized & $c_i$'s & $g_i$ tractable and conjugate to $G_0$ \\
 \hline
 \multirow{2}{5em}{Neal 5} & Original  & $c_i$'s, $\theta^*_h$'s & any $f$ and $G_0$ \\
 & Sanitized & $c_i$'s, $\theta^*_h$'s, $\tilde Y_{i,j}$'s & any $f$ and $G_0$ \\
 \hline 
 \multirow{2}{5em}{Slice} & Original  & $P$, $c_i$'s & any $f$ and $G_0$ \\
 & Sanitized & $P$, $c_i$'s, $Y_{i}$'s & any $f$ and $G_0$ \\
 \hline
 \end{tabular}
 \caption{State of the MCMC algorithms and the modeling assumptions needed when working on the original dataset or the sanitized ones. The $c_i$'s are cluster assignments, $\theta^*_h$'s are cluster-specific parameters, $P$ is the mixing measure, $Y_i$'s are confidential observations, and $\tilde Y_{i, j}$ is the $j$-th auxiliary variable corresponding to confidential observation $i$.}
 \label{tab:algo_table}
\end{table}
}

We propose two MCMC algorithms for posterior sampling from this model.  
The first is a marginal MCMC sampler that is an instance of the so-called pseudo-marginal MCMC, while the second is a conditional MCMC sampler based on missing data imputation.
\MBtext{Table \ref{tab:algo_table} shows a summary of the MCMC algorithms considered in this paper, the corresponding state of the chain, and modeling assumptions.}

\vspace{-1em}
\subsection{Marginal MCMC sampling: private Neal 5 algorithm}

Starting with \eqref{eq:bnp_mix}, we operate two marginalizations: first we integrate out $P$ by relying on the well-known generalized P\'olya urn characterization of the Dirichlet measure \citep{Fer73}, and then marginalize out the $Y_i$'s,  To this end, letting $\bm \theta^* = (\theta^*_1, \ldots, \theta^*_k)$ be the unique values in $\theta_1, \ldots, \theta_n$ and denoting by $n_h = \sum_{i=1}^n \indicator [\theta_i = \theta^*_h]$, we obtain
\begin{align}\label{eq:marg_marg_mix}
&Z_i \mid \theta_i  \ind g(\cdot \mid \theta_i) = \int_{\mathbb Y} Q_i(\cdot \mid y) f(y \mid \theta_i) \dd y, \qquad i=1, \ldots, n \\
&\prob(\bm \theta \in \dd \bm \theta)= \frac{\alpha^k}{(\alpha)_{(k)}}  \prod_{h=1}^k (n_h - 1)! G_0(\dd \theta^*_h).
\end{align}
The main challenge in sampling from the conditional distribution \eqref{eq:marg_marg_mix} is that $g(\cdot \mid \theta)$ is intractable, in the sense that it cannot be evaluated analytically. \cite{Ber21ABC} propose a general strategy to deal with intractable kernels in the context of BNP mixtures, proposing the use of approximate Bayesian computation. However, this might be inefficient for large datasets.
Instead, we develop an exact pseudomarginal MCMC sampler~\citep{And09} using an unbiased estimator $\hat g$ for the mixture kernel.

From \eqref{eq:marg_marg_mix}, it is clear that the conditional distribution of $\theta_i$ is
\begin{equation}\label{eq:xi_fullcond}
    \theta_i \mid \theta_{-i}, Z  \sim  b \pi(\theta_i \mid \theta_{-i}) g(Z_i \mid \theta_i),
\end{equation}
where $\theta_{-i}$ is $(\theta_1, \ldots, \theta_{i-1}, \theta_{i+1}, \ldots, \theta_n)$,  $b$ is a normalizing constant, and $\pi(\theta_i \mid \theta_{-i}):=  \sum_{j=1}^{n-1} \frac{1}{\alpha + n - 1} \delta_{\theta_j}(\cdot) +  \frac{\alpha}{\alpha + n - 1}  G_0(\cdot)$ is the conditional prior of $\theta_i$ given $\theta_{-i}$.
As in Algorithm 5 of \cite{neal2000markov}, we can sample from the conditional distribution \eqref{eq:xi_fullcond} by a Metropolis-Hastings step, where the proposal is $\theta^\prime_i \sim \pi(\cdot \mid \theta_{-i})$, and the acceptance rate is
\begin{align*}
    \min\left\{1; \frac{\pi(\theta_i \mid \theta_{-i}) g(Z_i \mid \theta_i)}{\pi(\theta^\prime_i \mid \theta_{-i}) g(Z_i \mid \theta^\prime_i)} \frac{\pi(\theta^\prime_i \mid \theta_{-i})}{\pi(\theta_i \mid \theta_{-i})} \right\} = \min\left\{1; \frac{g(Z_i \mid \theta_i)}{g(Z_i \mid \theta^\prime_i)} \right\}.
\end{align*}
Since $g(\cdot \mid \theta)$ cannot be evaluated analytically, we employ a pseudo-marginal MCMC step that only requires an unbiased estimator for $g(\cdot \mid \theta)$. Assuming without loss of generality that $Q_i$ has a density $q_i$ with respect to some dominating measure, a natural estimator is
\begin{equation}\label{eq:pseudo_prop}
    \hat g_i := \hat g(Z_i \mid \theta_i) = \frac{1}{m} \sum_{j=1}^m q_i(Z_i \mid \tilde Y_{i, j}), \qquad \tilde Y_{i, j} \sim f(\cdot \mid \theta_i). 
\end{equation}
The state space of the resulting sampler consists then of $(\theta_i, \hat g_i)$. Each iteration, and for each $i$, the sampler proposes new values $(\theta^\prime_i, \hat g^\prime_i)$ from equations \eqref{eq:xi_fullcond} and \eqref{eq:pseudo_prop} and accepts them with probability $\min\left\{1; \hat g^\prime_i / \hat g_i \right\}$.
Following Algorithm 2 in \cite{neal2000markov}, rather than storing the $\theta_i$'s, one can also just store their unique values $\theta^*_h$ and cluster allocations $c_i$, such that $\theta_i = \theta^*_h$ if and only if $c_i = h$. In this setting, it is possible to additionally update the $\theta^*$ directly to improve mixing. If this is needed, we must include the $\tilde{Y}_{i,j}$ in our state space (rather than just $\hat{g}$), with $p(\theta^*_h \mid \cdots) \propto \prod_{i: c_i = h} f(\{\tilde{Y}_{i,j}\} \mid \theta^*_h) G_0(\dd \theta^*_h)$.  Algorithm \ref{algo:neal5} details the pseudocode of our marginal sampling. 

For the setting of $\varepsilon$-DP, {we can easily extend the results in \cite{ju_data_2022} to our pseudo-marginal move obtaining  lower bound for the acceptance rate: 
\begin{proposition}\label{prop:lb}
    If the channel $q_i$ satisfies $\varepsilon$-DP as in \eqref{eq:dp_def} then 
    $\min\left\{1; \hat g^\prime_i / \hat g_i \right\} \geq e^{-\varepsilon}$
\end{proposition}
\begin{proof}

    The proof follows by  observing that for any pair $(\tilde Y^\prime_{i, j}, \tilde Y_{i, j})$, $\varepsilon$-DP imples 
    $
        q_i(Z_i \mid \tilde Y^\prime_{i, j}) > e^{-\varepsilon }q_i(Z_i \mid \tilde Y_{i, j})
    $, so that 
    $
        \sum_{j=1}^m (q_i(Z_i \mid \tilde Y^\prime_{i, j}) - e^{-\varepsilon }q_i(Z_i \mid \tilde Y_{i, j})) > 0.
    $
    Then, noting that
    $
        \frac{\hat g^\prime_i}{\hat g_i} = \frac{\sum_{j=1}^m q_i(Z_i \mid \tilde Y^\prime_{i, j})}{\sum_{j=1}^m q_i(Z_i \mid \tilde Y_{i, j})}
    $,
    we immediately have our result.
\end{proof}
In practice, we observe much larger acceptance rates than this theoretical lower-bound. 
Moreover, we also empirically observe that for $(\varepsilon, \delta)$-DP channels and $\rho$-ZCDP channels the acceptance rates are satisfactory, e.g., larger than $10\%$.

{
\spacingset{1.0}
\SetNlSty{textbf}{[}{]}
\begin{algorithm}[t!]
	\textbf{Input}:{ Sanitized data $Z_1, \ldots, Z_n$.}
 
    \textbf{Initialize}:{ $c_1, \ldots, c_n$, the unique values $\theta^*_1, \ldots, \theta^*_k$, and the $Y_i$'s}
	\DontPrintSemicolon

    \For{each MCMC iteration}{

        \For{$i = 1, \ldots, n$} {
        
            Sample $c_i^\prime \in \{1, \ldots, k+1\}$ with
            $
                \prob(c_i^\prime = h) \propto \begin{cases}
                    n^{-i}_h & \text { if } 1 \leq h \leq k \\
                    \alpha & \text { if } h = k+1
                \end{cases}
            $,  \\
            where $n^{-i}_h$ is $n_h$ computed when observation $i$ is removed from the state.

            \textbf{if} $c_i^\prime = k+ 1$, sample $\theta^*_{k+1} \sim G_0$
            
            Sample $\{\tilde Y_{i, j}^\prime\} \sim f(\cdot \mid \theta^*_{c_i^\prime})$. 

            With prob.\ $\min\left\{1; \hat g^\prime_i / \hat g_i \right\}$ set $(c_i, \{\tilde Y_{i, j}\}) = (c_i^\prime, \{\tilde Y_{i, j}^\prime\})$, else revert all changes.
        }

        \For{$h=1, \ldots, k$} {
            Sample $\theta^*_h$ from the density 
            $
                p(\theta^*_h \mid \cdots) \propto \prod_{i: c_i = h} f(\{\tilde{Y}_{i,j}\} \mid \theta^*_h) G_0(\dd \theta^*_h)
            $
        }
    }
	
	\textbf{end}
	\caption{\label{algo:neal5} Private Neal 5 algorithm}
\end{algorithm}
}

\vspace{-1.1em}
\subsection{Conditional MCMC sampling: private slice-sampling}\label{sec:slice_algo}

The development of a private version of the slice-sampling of \cite{kalli2011slice} is straightforward. From the model \eqref{eq:bnp_mix} with a Dirichlet measure $\mathscr Q$, we write the joint distribution
\begin{multline}\label{eq:slice_joint}
    \prob(Z \in \dd z, Y \in \dd y, \theta \in \dd \theta,  P \in \dd \tilde p) = \left\{\prod_{i=1}^n q_i(z_i \mid y_i) \dd z_i \, f(y_i \mid \theta_i) \dd y_i \, \tilde p(\dd \theta_i) \right\}  \mathscr Q(\dd \tilde p),
\end{multline}
from which all the full-conditional distributions for the random variables involved can be derived. Note that $(P, \theta)$ is independent of $Z$ given $Y$. Therefore, with respect to the original slice-sampling of \cite{kalli2011slice}, we must only additionally sample the $Y_{i}$'s from their full-conditional distributions. From \eqref{eq:slice_joint} the conditional density of $Y_i$ is
$    f_{Y_i}(y \mid \text{rest}) \propto  q_i(z_i \mid y) f(y \mid \theta_i),$
which can be easily sampled by a standard Metropolis-Hastings step. Following~\cite{ju_data_2022}, we used $f(\cdot \mid \theta_i)$ as the proposal proposal distribution, giving in the same mixing guarantee of Proposition \ref{prop:lb} for $\varepsilon$-DP, and resulting in good empirical performance. Having sampled the $Y_i$'s, the updates of $P$ and $\theta$ 
proceed as in \cite{kalli2011slice}, disregarding the $Z_i$'s.
Other conditional algorithms like \cite{papaspiliopoulos2008retrospective}, \cite{ishwaran2001gibbs}, \cite{canale2022importance}, or \cite{arbel2017moment}
may also be considered.

\subsection{Algorithms for tractable marginalized kernels}\label{sec:marg_kernels}

Whenever the marginalized kernel $g_i(z \mid \theta) = \int_{\mathbb Y} q_i(z \mid y) f(y \mid \theta) \dd y$ has a tractable analytic form, we can integrate out the $Y_i$'s and sample from the posterior distribution of the marginal model \eqref{eq:marg_marg_mix}, or its equivalent formulation when $P$ is not marginalized out, by using standard algorithms. This is the case, for instance, of a Gaussian mixture model for the $Y_i$'s coupled with the Gaussian DP mechanism, or the case of the Binomial DP mechanism in \cite{chen22} paired with a mixture of Beta kernels. We refer to \cite{dimitrakakis2017differential} for other examples. 

As an example, we consider the case of Gaussian mixtures with a Gaussian DP mechanism for univariate data. 
After marginalizing out $Y_i$'s, we have that the kernel $g$ is
\begin{align*}
    g(\cdot \mid \theta_i) = \int_{\R \times \R_+} \mathcal N(\cdot \mid y, \eta^2) \mathcal N(y \mid \mu_i, \sigma^2_i) \dd y = \mathcal N(\cdot \mid \mu_i, \eta^2 + \sigma^2_i).
\end{align*}
Accordingly, the Bayesian model for the $Z_i$'s, with the $Y_i$'s marginalized out, is
\begin{equation}\label{eq:diff_gauss_marg1}
        Z_i \mid \mu_i, \sigma^2_i \ind \mathcal N(\mu_i, \eta^2 + \sigma^2_i), \quad
        \mu_i, \sigma^2_i \mid P  \iid P,  \quad 
        P \sim DP(\alpha G_0).
\end{equation}
The posterior distribution under \eqref{eq:diff_gauss_marg1} can be easily sampled. Note that standard choices for $G_0$, such as the Normal-Inverse-Gamma distribution, are not conjugate to $g(\cdot \mid \cdot)$, and one must resort to, e.g., Algorithm 8 in \cite{neal2000markov}. However, by means of a careful choice of the prior distribution, we show it is possible to marginalize out the unique values $(\mu^*_h, \tau^*_h)$, obtaining an MCMC algorithm that samples only the cluster allocations. Setting $\tau_i^2 = \eta^2 + \sigma^2_i$, which is a random variable supported on $[\eta^2, +\infty)$, we have
\begin{equation}\label{eq:diff_gauss_marg}
        Z_i \mid \mu_i, \tau^2_i  \ind \mathcal N(Y_i, \tau^2_i), \quad 
        \mu_i, \tau^2_i \mid P \iid P,  \quad
        P \sim DP(\alpha G_0).
\end{equation}
By choosing $G_0(\dd \mu, \dd \tau^2) = \mathcal N(\mu \mid \mu_0,\tau^2 / \lambda) IG(\tau^2 \mid a, b) \indicator_{[\eta^2, +\infty)}(\tau^2) \dd \mu \, \dd \tau^2$ we obtain a conjugate model. Now \eqref{eq:diff_gauss_marg} can be sampled by means of  Algorithm 3 in \cite{neal2000markov}, or by means of the split and merge algorithm in \cite{jain2004split}, which usually leads to a better mixing. Critical to this derivation is (i) the choice of an appropriate perturbation channel that allows for analytic convolution with respect to the mixture kernel, and (ii) the choice of $G_0$ that takes into account the mixture kernel as well as the privacy constraint. The main difference with the standard Gaussian mixture with a Normal-Inverse-Gamma centering measure is that the marginal distribution of data inside a cluster is not a Student t distribution. Nonetheless its density function can be computed by Bayes' formula.

\vspace{-.1in}
\section{Numerical Illustrations}\label{sec:num}
\vspace{-.1in}
We present two simulation studies to assess the performance of our marginal and conditional sampling schemes across different privacy mechanisms, privacy levels and sample sizes. 
To evaluate samplers, we monitor the effective sample size (ESS) of the number of clusters, i.e. the number of unique values across the $c_i$'s. We note that when varying the privacy level $\varepsilon$, the posterior distribution is itself changing, and it makes little sense to compare effective sample sizes across different posteriors. In this case, we monitor the acceptance rate of the update of the $Y_i$'s, allowing us to evaluate how tight the bound in~\Cref{prop:lb} is. We note this is not directly comparable across samplers, since for the marginal sampling scheme the $Y_i$'s are updated together with the $c_i$'s, while for the conditional sampling scheme these updates are done separately. 
Unless otherwise specified, chains are run for 100,000 iterations, discarding the first 50,000 as burn-in steps, and we report the median over 50 independent repetitions of the acceptance rates and ESS.

\subsection{Private Neal 5 algorithm versus private slice-sampling}\label{sec:ex_lap}

We simulate $n=50, 200, 500, 1000$ data points from a mixture of three equally weighted truncated Gaussian distributions on the interval $I = [-10, 10]$, with means $-5, 0, 5$ and variance equal to one. We sanitize data through the Laplace mechanism, for $\varepsilon = 1.0, 2, 5, 10, 50$. We model the confidential data with a location-scale mixture of Gaussians, i.e. $\theta_i = (\mu_i, \sigma^2_i)$ and $f(y \mid \theta_i)$ is the Gaussian density function with mean $\mu_i$ and variance $\sigma^2_i$.  The base measure $G_0(\mu, \sigma^2)$ is assumed to be the Normal-Inverse-Gamma distribution $\mathcal N (\mu \mid 0, 10 \sigma^2) \times IG(\sigma^2 \mid 3, 3)$. Note that despite the confidential data being compactly supported, for simplicity, we assume a mixture model with support on the whole set $\R$.

\begin{figure}[t]
    \centering
    \includegraphics[width=\linewidth]{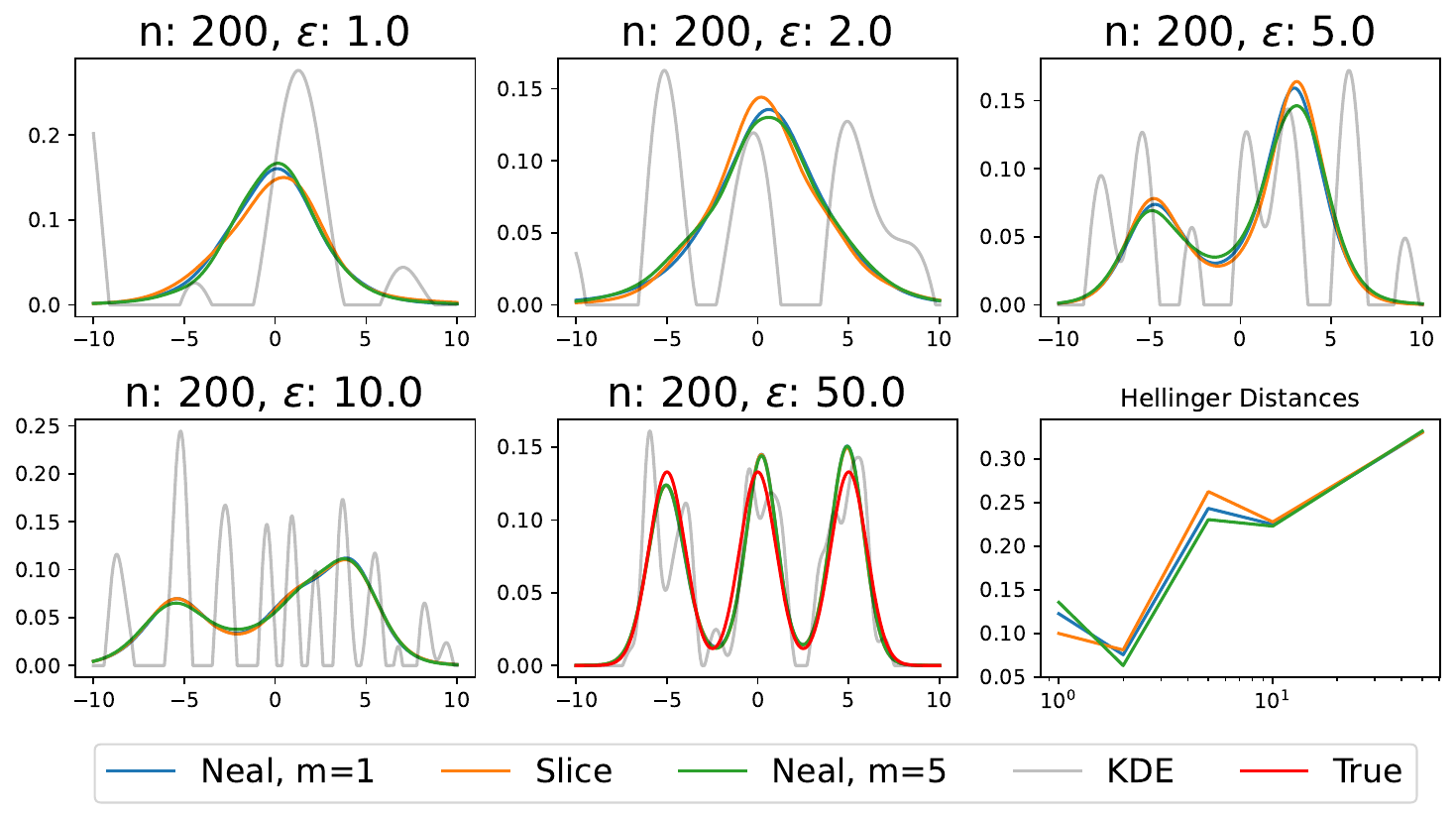}
    \caption{First five plots: density estimate for different privacy levels $\varepsilon$  when $n=200$. Last plot: Hellinger distance between posterior and prior mean density versus $\varepsilon$. (Section \ref{sec:ex_lap}).}
    \label{fig:laplace_dens_estimate_n200}
\end{figure}

We compare three proposed algorithms, the private Neal 5 algorithm with {$m=1$ auxiliary variables, with $m=5$ auxiliary variables}, and the private slice-sampler. \MBtext{We also consider a non-Bayesian alternative recently proposed in \cite{farokhi2020deconvoluting} which estimates the density function of the $Y_i$'s by the deconvolution kernel density estimator of \cite{stefanski1990deconvolving}. See also \cite{delaigle2004practical} and references therein. To compute the estimator we use the R package \texttt{decon} \citep{wang2011deconvolution}, which applies a bootstrap approach to select the optimal kernel bandwidth.}
Figure \ref{fig:laplace_dens_estimate_n200} shows the posterior density estimate for the $Y_i$'s when $n=200$, that is the posterior expectation of $\int f(\cdot \mid \theta)P(\dd \theta)$. 
\MBtext{The density estimates agree for all algorithms, and therefore in Figure \ref{fig:laplace_dens_estimate_full} in the Appendix we report only the Slice sampling algorithm and the associated 95\% pointwise credible bands (which can only be obtained by \emph{conditional} algorithms). }
We note how for strong privacy levels such as $\varepsilon = 1, 2$ \MBtext{(corresponding to large additive noise)}, the density estimate is poor, even with $n=1,000$ observations. For lower privacy levels, such as $\varepsilon = 50$, even $n=50$ observations are sufficient for reasonable density estimates. 
We arrive at similar conclusions when we focus on clustering performance, as measured by the Adjusted Rand Index (ARI) between the true and estimated clustering structure of the $Y_i$'s  (see Figure \ref{fig:lap_ari} in the appendix).
\MBtext{The deconvolution KDE performs rather poorly in all situations. 
One interesting difference between the Bayesian and frequentist estimates is that in high privacy regimes, the large amount of added noise results in the Bayesian estimates reverting to the prior predictive density, as shown by the last plot in Figure \ref{fig:laplace_dens_estimate_n200}. This gives an easy-to-check condition to assess the 
informativeness of the observed dataset under the privacy constraints.
On the other hand, without additional regularization, the frequentist estimates are all over the place when the noise is significant, and there is no such criterion to decide whether or not to trust the density estimates.}

\begin{figure}[t]
\begin{minipage}[c]{0.7\textwidth}
\centering
    \begin{subfigure}{\linewidth}
        \centering
        \includegraphics[width=\linewidth]{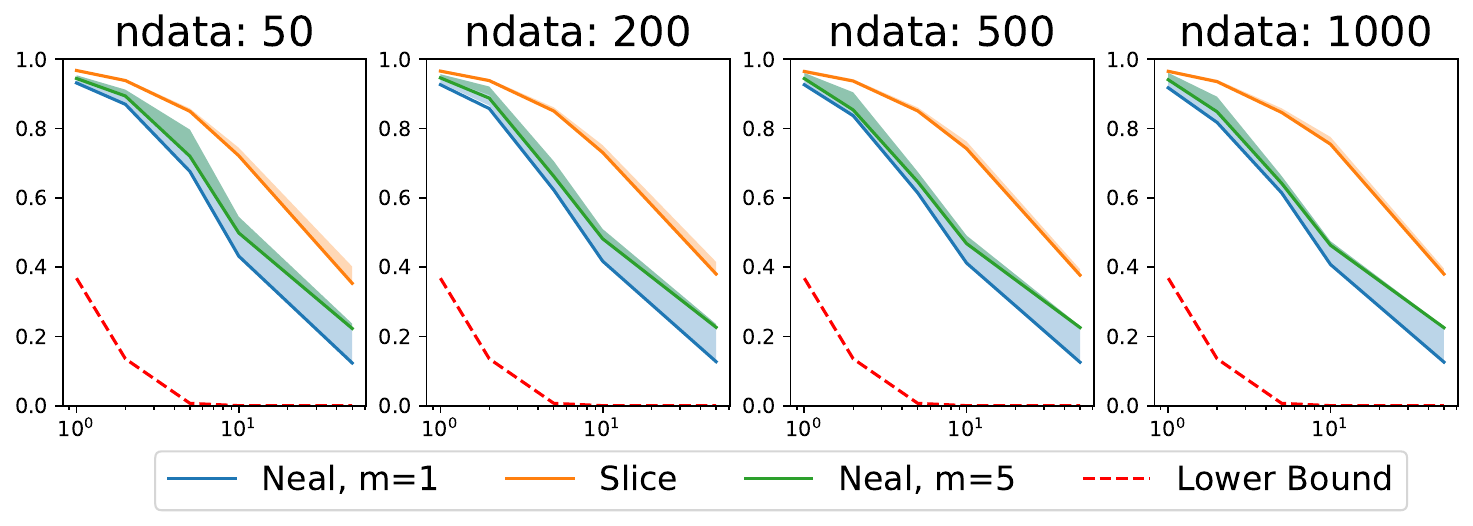}
    \end{subfigure}
    \begin{subfigure}{\linewidth}
        \centering
        \includegraphics[width=\linewidth]{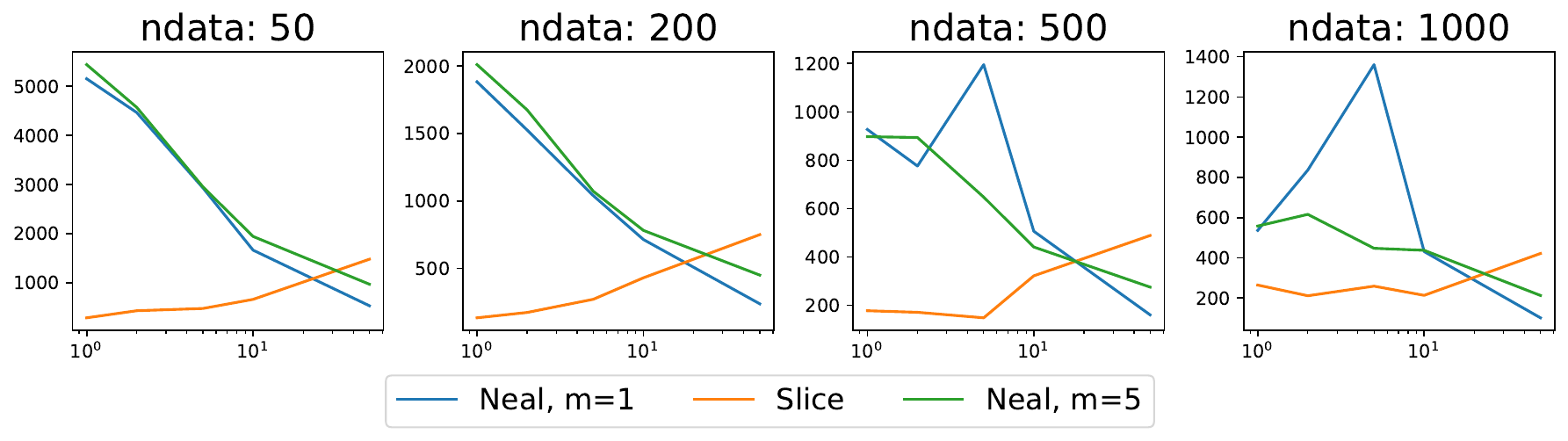}
    \end{subfigure}
    \end{minipage}
\begin{minipage}[c]{0.28\textwidth}
\vspace{-.1in}
    \caption{Top: acceptance rate as a function of the privacy level (the dashed line is the theoretical lower bound). 
    Bottom: effective sample size for the number of clusters.}
    \label{fig:laplace_ess}
    \end{minipage}
\end{figure}

Figure \ref{fig:laplace_ess} shows the median acceptance rate and ESS across the 50 replicates. As expected, the acceptance rates decrease with $\varepsilon$, but are still significantly larger than the theoretical lower bound $\exp(-\varepsilon)$ in \Cref{prop:lb}. For the private Neal 5 algorithm, the ESS for the number of clusters also decreases with $\varepsilon$, since the update of the cluster allocations is intimately tied to the update of the private observation. On the other hand, for the private slice-sampling, the ESS increases with $\varepsilon$, likely due to the decoupling of the sampling of the private data from the remaining of the parameters in the MCMC sampling. We also see that for lower values of $\varepsilon$, the private Neal 5 algorithm yields higher ESSs than the private slice-sampling. \MBtext{The use of $m=1$ or $m=5$ in the Neal algorithm does not significantly affect the ESS or the density estimate, though, as expected, increasing $m$ slightly reduces the variability of the estimates. However this increases the computational cost by almost a factor $m$.}
The runtimes are essentially unchanged when varying $\varepsilon$, and comparable for Neal 5 algorithm and the slice-sampler for all settings except when $n=1000$, when the slice-sampler is approximately ten times faster. The slice-sampler could be further accelerated by using parallel computations, since all the updates are easily parallelizable, in contrast to Algorithm \ref{algo:neal5} which is inherently sequential. 

\subsection{Laplace mechanism versus wavelet-based mechanism}\label{sec:ex_wav}

\begin{figure}[t]
    \centering
    \includegraphics[width=\linewidth]{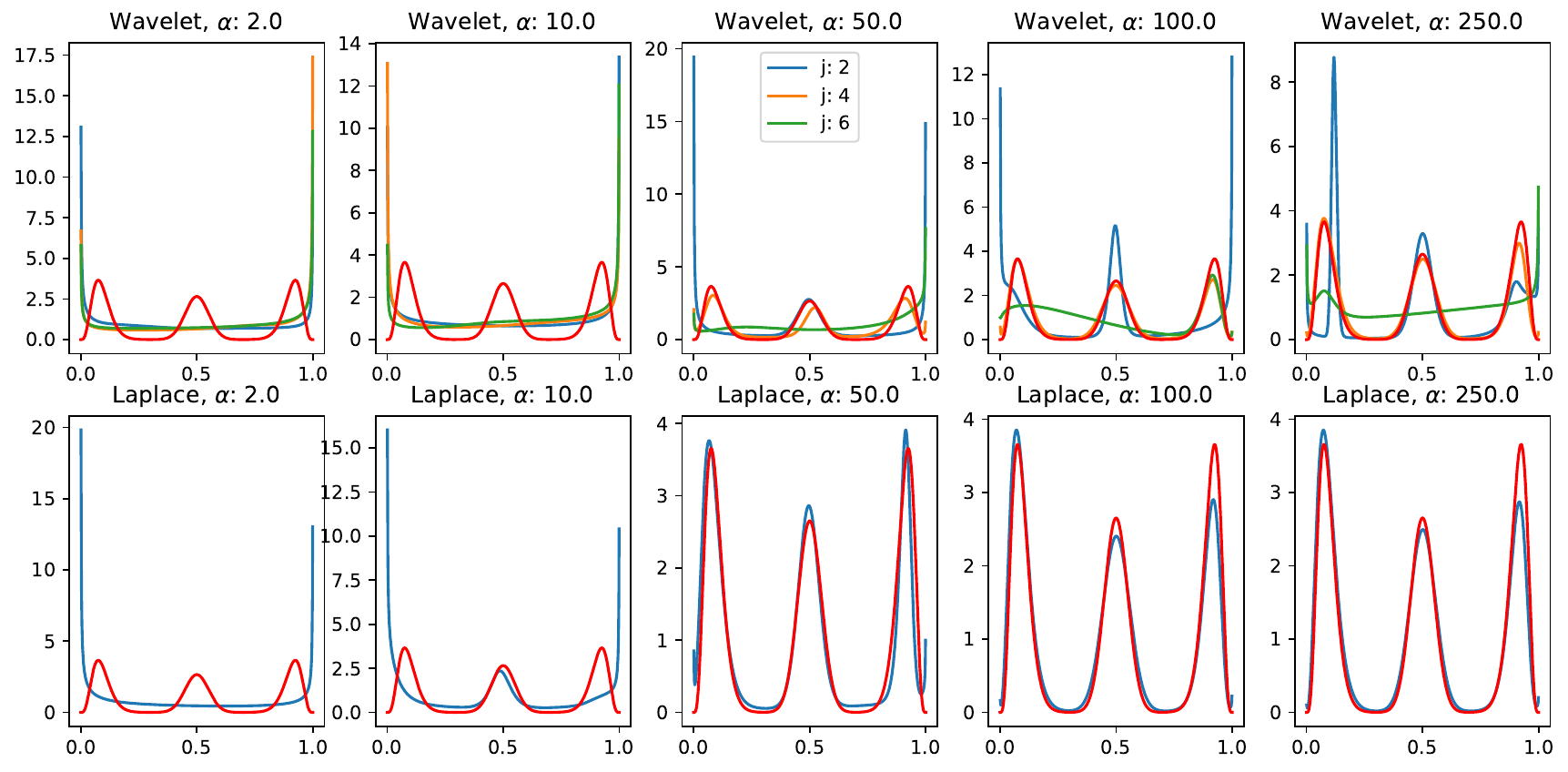}
    \caption{Density estimate for the simulated example in Section \ref{sec:ex_wav}. Top row: wavelet privacy mechanism. Bottom row: Laplace privacy mechanism.}
    \label{fig:wav_dens}
\end{figure}

\MBtext{In this section, we demonstrate the applicability of our ideas to more general privacy mechanisms, considering specifically the wavalet-based mechanism of \cite{butucea2020local}. See Appendix \ref{app:background} for further details on this mechanism. We compare posterior inference under this and the Laplace mechanism for the same privacy levels.}

We simulate $n=250$ data points from a mixture of three equally weighted Beta distributions with parameters $(5, 50)$, $(50, 50)$ and $(50, 5)$, respectively. We assume a Dirichlet process mixture of Beta kernels with parameters $a > 0 , b > 0$ for the $Y_i$'s, i.e.\ $f(\cdot \mid \theta) = \mbox{Beta}(\cdot \mid a, b)$, and set the base measure as $G_0(\dd a, \dd b) = \mbox{Gamma}(\dd a \mid 2, 2) \mbox{Gamma}(\dd b \mid 2, 2)$. 
See Appendix \ref{app:numerics} for further details. Below, we report posterior inference obtained using the private Neal 5 algorithm, though similar results hold for the private slice-sampler.

We consider privacy levels $\varepsilon = 2, 10, 50, 100, 250$ and, when adopting the wavelet-based mechanism, we consider $J= 2, 4, 6$. Figure \ref{fig:wav_dens} reports the posterior estimates of the density function, and Figure \ref{fig:wav_arate} in the Appendix the corresponding acceptance rates. We see that when $J=2$ the density estimate is rather poor: for $\varepsilon = 2, 10, 50$ the estimated density does not capture the three modes and pushes all the mass to the boundaries of the interval, when $\varepsilon=50$ it captures the mode in $1/2$ but still fails to properly capture the boundaries. When $J=6$, the acceptance rates are rather low when $\varepsilon \geq 100$. The density estimates are poor as well: for all values of $\varepsilon$ the estimated density is very different from the true one. When $J=4$ instead, the estimated density matches the true one for $\varepsilon \geq 50$. Similarly, when using the Laplace mechanism, the estimated densities match the true one for $\varepsilon \geq 50$.

\MBtext{Our experiments show that posterior estimates inherit the sensitivity of the wavelet-based mechanisms to their hyperparameters. As we vary $J$ between 2 and 6, we found that posterior density estimates are consistent with the true data generating distribution only when setting $J=4$. In contrast, when adding Laplace noise, density estimates are sensible.
}

\vspace{-.2in}
\subsection{Private Neal 5 versus non-private (marginal) algorithms}\label{sec:ex_gauss}

Consider the BNP mixture model \eqref{eq:bnp_mix} where $Q_i(\cdot \mid Y_i) = \mathcal N(\cdot \mid Y_i, \eta^2)$.
Following \cite{dwork_roth} we set $\eta = \sqrt{2 \log \frac{1.25}{\delta}} \frac{\Delta}{\varepsilon}$, but other choices can be considered to ensure $(\varepsilon, \delta)$-DP \citep{balle_wang, zhao2019gaussdp}.  We consider the same data generating process as in Section \ref{sec:ex_lap} and let $\delta = \{0.01, 0.1, 0.25\}$, $\varepsilon = \{5, 10, 25, 50\}$. For smaller values of $\varepsilon$, the variance $\eta^2$ is extremely large (e.g., when $\varepsilon = 0.5, \delta=0.01$, $\eta^2 > 3800$) thus making density estimation essentially impossible.
We simulate $n=250$ observations and compute posterior inference under model \eqref{eq:diff_gauss_marg} using either Neal's algorithm 2 or algorithm 3, which marginalizes w.r.t. the unique values $(\mu^*_h, \tau^*_h)$, and under model \eqref{eq:marg_marg_mix} using the private Neal 5 algorithm with $m=1$. 

Figure \ref{fig:gauss_dens_delta01} shows the posterior density estimates for $\delta=0.1$, see Figure \ref{fig:gauss_dens} in the Appendix for the other cases. 
\MBtext{The three algorithms produce essentially identical density estimates, which are rather poor for the smallest values of $\varepsilon$ and $\delta$ considered here, excellent when $\varepsilon = 50$ across all values of $\delta$ and somewhat satisfactory when $\varepsilon = 10, 25$.}
Figure \ref{fig:gauss_ess} instead shows the effective sample sizes of the number of clusters. Neal 2 and Neal 3 algorithms outperform the private Neal 5 algorithm, and Neal 3 algorithm is always better than Neal 2 algorithm. Such a behaviour was expected, since these algorithms marginalize over the $Y_i$'s and Neal 3 algorithm further marginalizes over the unique values. Moreover, while the effective sample size for the private Neal 5 algorithm decreases with the privacy level $\varepsilon$, as already observed in Section \ref{sec:ex_lap}, this does not happen for Neal 2 and Neal 3 algorithms.

\begin{figure}[t]
    \centering
    \includegraphics[width=0.9\linewidth]{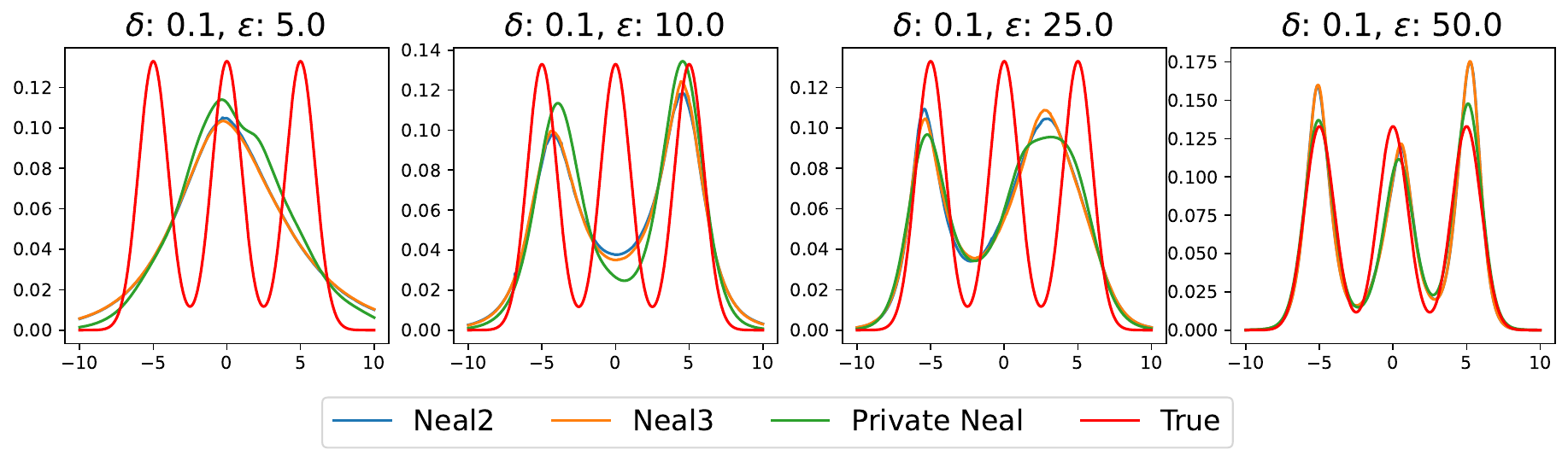}
    \caption{Density estimate for the example in Section \ref{sec:ex_gauss} for different $\varepsilon$ and $\delta=0.1$.}
    \label{fig:gauss_dens_delta01}
%
    \centering
    \includegraphics[width=0.8\linewidth]{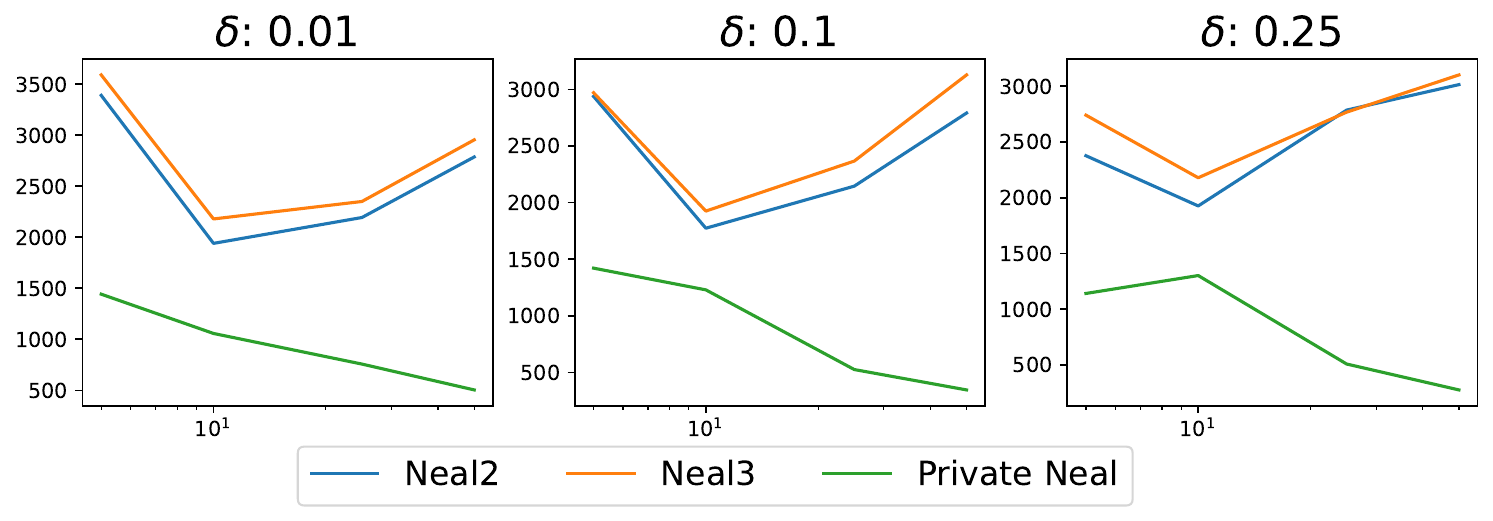}
    \caption{Effective sample size of the number of clusters for the example in Section \ref{sec:ex_gauss} }
    \label{fig:gauss_ess}
\end{figure}

\vspace{-.2in}
\section{BNP mixture modeling under global DP}\label{sec:global}
\vspace{-.1in}
We extend of our framework to the setting of global DP, whereby the privacy mechanism is not applied observation-wise but to the whole dataset $Y = (Y_1, \ldots, Y_n)$. $Y$ is now stored with a trusted data holder, who then releases the sanitized $W \sim Q(\cdot \mid Y)$ where $Q(\cdot \mid \cdot)$ is the privacy channel. Under $\varepsilon$-DP, $Q$ must satisfy
\[
    \sup_A \sup_{y, y^\prime: d_H(y, y^\prime) = 1} \frac{Q(W \in A \mid Y = y)}{Q(W \in A \mid Y = y^\prime)}   \leq e^{\varepsilon}
\]
where the first supremum is over all measurable sets and $d_H$ is the Hamming distance, i.e.\ $d(y, y^\prime) = |\{i: y_i \neq y_i^\prime\}|$. Extensions to $(\varepsilon, \delta)$-DP and $\rho$-zCDP are straightforward.

A few global privacy mechanisms for density estimation have been proposed in \cite{wasserman2010statistical}. We consider the ``sampling from a smoothed histogram'' mechanism, for $Y_i \in [0, 1]$ which releases an i.i.d. sample $W = (W_1, \ldots, W_k) \in \mathbb Y^k$ from the density 
\begin{equation}\label{eq:smooth_hist}
   f_{m, \delta}(x) := (1 - \delta) \sum_{j=1}^m \frac{c_j}{n} \indicator_{B_j}(x) + \delta, 
\end{equation}
where $\{B_1, \ldots, B_m\}$ is a partition of $[0, 1]$, $c_j = \sum_{i=1}^n I(Y_i \in B_j)$, and $n$, $m$, $\delta$ and $\varepsilon$ satisfy 
$    k \log \left(\frac{1 - \delta}{\delta} \frac{m}{n} - 1 \right) \leq \varepsilon.$
Other alternatives 
include directly releasing a perturbed histogram as proposed in \cite{dwork2006differential}, or the projection estimator in \cite{lalanne2023cost}.
Importantly, these perturbation mechanisms yield a channel $Q(\cdot \mid Y)$ whose density is easily computable. On the other hand, the ``sampling from a perturbed mechanism'' or the mechanism based on 
the truncated perturbed Fourier expansion proposed in \cite{wasserman2010statistical} yield to channels $Q(\cdot \mid Y)$ with intractable densities which makes posterior inference cumbersome.

Under the global DP setting for density estimation, the BNP mixture model is essentially identical to \eqref{eq:bnp_mix}, where we replace the first line with $W \mid Y \sim Q(\cdot \mid Y)$. A significant difference between the local DP setting and global DP setting is that global DP does not typically allow sensible inferences on the latent clustering, since $W$ does not preserve information about individual observations. 
\vspace{-.2in}
\subsection{Marginal and conditional MCMC sampling}\label{sec:global_algo}
We consider the joint distribution of data and parameters under global DP setting: 
\begin{multline}\label{eq:slice_global_joint}
    \prob(W \in \dd w, Y \in \dd y, \theta \in \dd \theta,  P \in \dd \tilde p) = Q(\dd w \mid Y) \left\{ \prod_{i=1}^n f(y_i \mid \theta_i) \dd y_i \, \tilde p(\dd \theta_i) \right\} \mathscr Q(\dd \tilde p).
\end{multline}
It is straightforward to extend the private slice-sampling to the global DP setting. The only difference from the local DP setting is that the conditional distribution of $Y_i$ is
\[
    f_{Y_i}(y \mid \text{rest}) \propto Q(\dd w \mid Y) f(y_i \mid \theta_i),
\]
from which we can sample by a standard Metropolis-Hastings step with $f(\cdot \mid \theta_i)$ as proposal, similarly to \Cref{sec:slice_algo}. Depending on the choice of the channel $Q$, this might become inefficient since evaluating the full-conditional density might scale as $\mathcal O(nk)$.  However, for the mechanism we consider, updating $Y_i$ changes at most one of the counts $c_j$ in \eqref{eq:smooth_hist} and thus we do not need to re-compute the histogram $f_{m, \delta}$ at every step. 

We can also design a marginal MCMC sampler for global DP, analogous to what we have done for local DP. In particular, by integrating out $P$ in \eqref{eq:slice_global_joint}, we obtain
\begin{equation*}
    \prob(W \in \dd w, Y \in \dd y, \theta \in \dd \theta) =  Q(\dd w \mid Y) \left\{ \prod_{i=1}^n f(y_i \mid \theta_i) \dd y_i \right\} m(\dd \theta),
\end{equation*}
where $m(\dd \theta) = \frac{\alpha^k}{(\alpha)_{(k)}}  \prod_{h=1}^k (n_h - 1)! G_0(\dd \theta^*_h)$ is the marginal distribution of a sample from a Dirichlet process. Taking inspiration from Algorithm \ref{algo:neal5}, here we propose to update $(Y_i, \theta_i)$ jointly. Writing  $\pi(\theta_i \mid \theta_{-i})$ for the conditional prior of $\theta_i$ given $\theta_{-i}$ as in \eqref{eq:xi_fullcond}, we obtain 
\begin{align*}
    & \prob(Y_i \in \dd y_i, \theta_i \in \dd \theta_i \mid \text{rest}) \propto Q(\dd w \mid Y) f(y_i \mid \theta_i) \dd y_i \pi(\theta_i \mid \theta_{-i}).
\end{align*} 
We sample from this using a Metropolis-Hastings step with a proposal $\prob(Y_i^\prime \in \dd y^\prime, \theta_i^\prime \in \dd \theta^\prime) =  \pi(\theta^\prime \mid \theta_{-i}) f(y^\prime \mid \theta^\prime_i) \dd y^\prime$ and acceptance probability
$    \min\left\{1, \frac{q(W \mid Y_1, \ldots, Y_i^\prime, \ldots, Y_n)}{q(W \mid Y_1, \ldots, Y_i, \ldots, Y_n)}\right\}$.
Here $q(\cdot \mid \cdot)$ is the density of the channel $Q$. Further considerations of marginal MCMC sampling schemes under global DP follow along the same lines discussed for local DP.

In Appendix \ref{app:global_simu}, we report a simulation where we compare the posterior inference under the privacy mechanism in \ref{eq:smooth_hist} using the two MCMC algorithms described above. We obtain the same conclusions of \Cref{sec:ex_lap}, namely both algorithms present much higher acceptance rates than the theoretical lower bound. The density estimates show a behavior similar to what is observed for the wavelet-based perturbation: one carefully needs to choose the hyper-parameters involved in $Q$ to get good estimates.

\vspace{-.2in}
\section{Real data analysis}\label{sec:blood_donors}
\vspace{-.1in}
We consider a dataset of blood donors at the Milano Department of the Associazione Volontari Italiani del Sangue (AVIS) between January 1st, 2010 and May 15th, 2016. Each data point is the log of the time between the first and second blood donation of one of $n=100$ individuals, so that $Y_i = \log(\#$ {number of days between first and second donation$)$.
Due to health concerns, individuals are not allowed to donate blood for 90 days after their donation. Moreover, the timespan in which data are collected ensures a maximum value for $Y_i$. Together, these imply $Y_i \in [4.49, 7.6]$, so that $\Delta = 3.11$. We refer to \cite{argiento2022clustering} for further analyses on the blood donors data.

We analyze this dataset under zero concentrated differential privacy, fixing a privacy loss budget $\rho = 17.8$ following the US Census Bureau standard for data referring to individuals \citep{Garfinkel2022Differential}. To obtain sanitized data, we consider \MBtext{two locally differentially private mechanisms, namely the Laplace mechanism (with scale $\Delta / \sqrt{2 \rho}$) and the Gaussian mechanism (with variance $\Delta^2 / (2 \rho)$). The resulting perturbed data are displayed in Figure \ref{fig:blood-pert} in the Appendix.}
We assume a Gaussian mixture model for the $Y_i$'s, with $G_0$ the Normal-Inverse-Gamma distribution: $G_0(\dd \mu, \dd \sigma^2) = \mathcal N(\dd \mu \mid \mu_0, 10 \sigma^2) \times IG(\dd \sigma^2 \mid a, b)$. The hyperparameters for $G_0$ are chosen in an empirical Bayes fashion, with $\mu = n^{-1} \sum_{i=1}^n Z_i$, and $(a, b)$ are chosen so that $\E[\sigma^2]$ is equal to one half of the empirical variance of the $Z_i$'s and $\mbox{Var}[\sigma^2] = 0.5$.
When using the Gaussian mechanism, we consider the private Neal 5 algorithm, the private slice-sampling, and the private Neal 2 and Neal 3 algorithms as discussed in \Cref{sec:marg_kernels}. When using the Laplace mechanism, we consider the private Neal 5 algorithm and the private slice-sampling. \MBtext{We also compare the output of the frequentist deconvolution estimator from Section \ref{sec:ex_lap}.}


The density estimates are reported in Figure \ref{fig:blood_dens} (top). The Bayesian estimates look essentially identical across both perturbation mechanisms and for any of the proposed algorithms, reflecting the importance of explicitly modeling the confidential data as well as the privacy-mechanism. Instead, while the deconvolution kernel density estimator  looks similar to the Bayesian estimate under the Gaussian mechanism, under the Laplace mechanism, it seems to be overfitting the data. We also compare the posterior of the mixture density given the privatized data and given the true non-sanitized data. Figure \ref{fig:blood_dens} (bottom) shows the posterior density estimates and the 95\% pointwise credible bands. It is clear how the density estimates are essentially identical when considering the privatized or the original dataset. On the other hand, as one would expect, the credible bands are thinner when considering the original observations.

\begin{figure}
\begin{minipage}[c]{0.62\textwidth}
\centering
        \includegraphics[width=\linewidth]{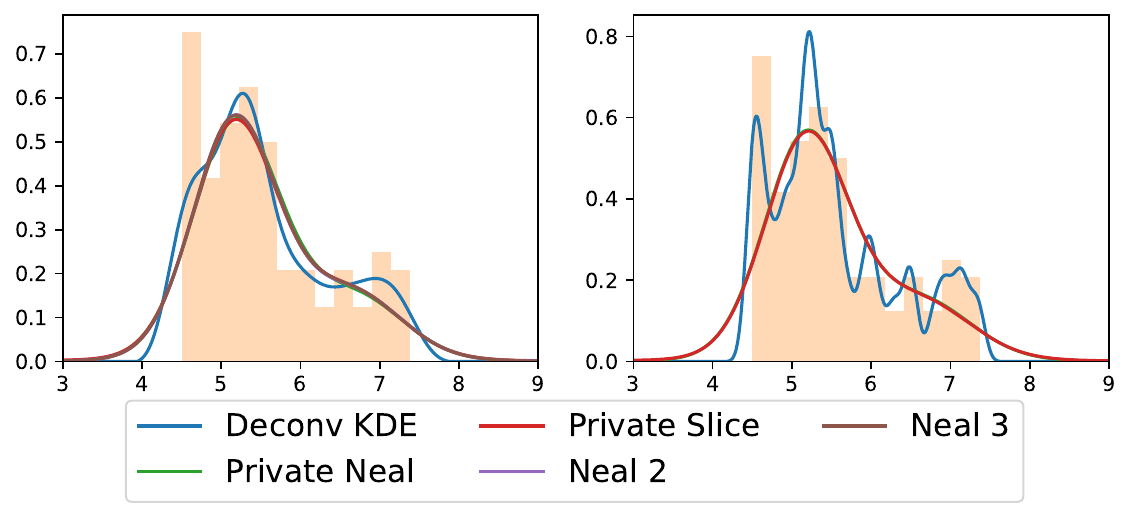}
        \includegraphics[width=\linewidth]{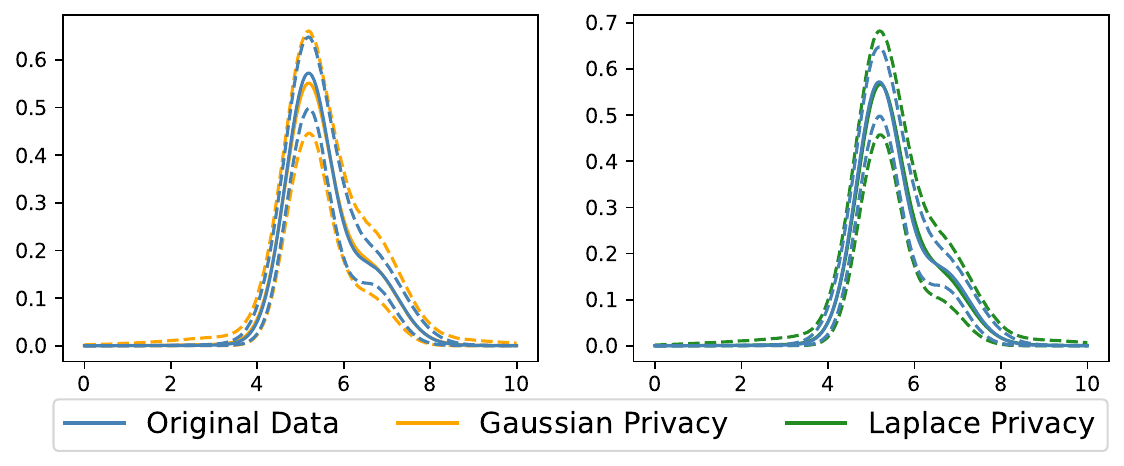}
\end{minipage}
\begin{minipage}[c]{0.37\textwidth}
    \caption{Density estimates for the AVIS dataset in \Cref{sec:blood_donors}. 
    Left plots: Gaussian perturbation. Right plots: Laplace perturbation.
    The top row shows posterior point estimates and the deconvolution KDE estimator.
    The bottom row shows the point estimates and pointwise 95\% credible bands. The histograms refer to the non-privatized data $Y_i$.}
    \label{fig:blood_dens}
\end{minipage}
\end{figure}

\vspace{-.2in}
\section{Discussion}\label{sec:disc}
\vspace{-.1in}
We considered the problem of Bayesian density estimation under differential privacy, providing a marginal and a conditional algorithm for posterior inference.
These algorithms are amenable to any perturbation mechanism that has a tractable density.
We also showed how in some cases, a careful choice of perturbation mechanism, mixture kernel, and base measure allow for several latent parameters to be marginalized out,  obtaining more efficient algorithms.
While we have focused on Dirichlet process mixtures  in the exposition, our algorithms are general regarding both the mixing measure and the mixture kernels, and using the \texttt{BayesMix} library, it is already possible to consider different modeling choices.

\MBtext{Compared to frequentist alternatives for differentially private density estimation, our proposal inherits the flexibility of the Bayesian framework, providing a natural way to quantify uncertainty and allowing easy extension to more complex models that embed our current model. 
Moreover, as a byproduct of our approach, we also obtain posterior cluster estimates.
From our simulations, it is clear that when too much noise is injected into the data, density estimation and clustering become infeasible.

Our algorithms are not specific to the setting of differential privacy, and can be applied to general deconvolution problems.
In particular, in the case of the Laplace mechanism we recover the setting in \cite{rousseau_wasserstein_2021}, where theoretical guarantees on the convergence of the posterior distribution are provided. Our examples in Section \ref{sec:ex_lap} can be seen as an empirical confirmation of their theoretical work.
As mentioned in~\citet{ju_data_2022}}, the efficiency of our algorithms are linked to the special structure of differential privacy, and it is interesting to see if other deconvolution problems possess structure that can be exploited similarly.
In the other direction, extending the analysis in \cite{rousseau_wasserstein_2021} to other kind of privacy mechanisms is a challenging and interesting task.
Another question that is worth pursuing is to more formally study the effect that the inclusion of a privacy mechanism has on the mixing time of the MCMC algorithms.

\section*{Acknowledgements}

Mario Beraha acknowledges the support by MUR, grant Dipartimento di Eccellenza 2023-2027.
Mario Beraha and Stefano Favaro received funding from the European Research Council (ERC) under the European Union’s Horizon 2020 research and innovation programme under grant agreement No 817257. Stefano Favaro gratefully acknowledges the financial support from the Italian Ministry of Education, University and Research (MIUR), “Dipartimenti di Eccellenza” grant 2018-2022.

\bibliographystyle{chicago}
\bibliography{privacy}

\appendix

\section{Background Material}\label{app:background}

\subsection{Wavelet-Based Perturbation and Estimators}

The wavelet-based perturbation mechanism of \cite{butucea2020local} considers a collection of wavelet functions
\[
    \psi_{j, k}(x) = \begin{cases}
    \varphi(x - k) \quad &\text{if } j=-1 \\
    2^{j/2}\psi(2^j x - k ) \quad &\text{if } j\geq 0,
    \end{cases}
\]
with $\varphi$ and $\psi$ the father and mother wavelet, respectively. We assume that the $Y_i$'s, $\varphi$, and $\psi$ are supported on $[0,1]$, though any other compact interval of $\mathbb{R}$ is applicable.  \cite{butucea2020local} consider $Z_i = (Z_{i, j, k})$ for $j=-1, \ldots, j_1$ and $k \in \mathbb Z$ such that
\begin{equation}\label{eq:wav_priv_full}
    Z_{ijk} \mid Y_i \sim \mathcal{L}(\psi_{j, k}(Y_i), s_j)
\end{equation}
for a suitable choice of the parameter $s_j$. They then consider the density estimator
\[
    \hat f(x) = \sum_{j=-1}^{j_1} \sum_{k \in \mathbb Z} \hat \beta_{jk} \psi_{j,k}(x), \quad \hat \beta_{jk} = \sum_{i=1}^n Z_{i, j, k},
\]
where $j_1$ controls the complexity of the estimator, proving that $ \hat f(x)$ achieves the minimax optimal rate for the $L_p$ loss over certain classes of densities. Because of  the compact support of the wavelet functions, for a fixed $j$, $\psi_{j, k}$ is not identically zero over $[0, 1]$ if and only if $k \leq 2^j$. 
In our work, we will consider \eqref{eq:wav_priv} where $\psi_{j, k}$ is the Haar wavelet: $\varphi(x) = \indicator_{[0, 1]}(x)$ and $\psi(x) = \indicator_{[0, 1/2)}(x)  - \indicator_{[1/2, 1)}(x)$. Under these assumptions, $\varepsilon$-DP is achieved if $s_{-1} = 12 / \varepsilon$ and $s{j} = \frac{12}{\varepsilon} \frac{\sqrt{2}}{\sqrt{2} - 1} 2^{j_1/2}$, see \cite[Equation 3.1]{butucea2020local}.
Note that $\psi_{-1, k}(x) \equiv 1$ for any $x \in [0, 1]$ so this is effectively useless and we will consider only $\psi_{j, k}$ for $j \geq 0$.

\section{Further Details on the Experiments}\label{app:numerics}

\subsection{Laplace versus Wavelet}

In this example, the mixture kernel $f(\cdot \mid \theta)$, $\theta=(a, b)$ is not conjugate to the base measure $G_0$. The corresponding full-conditional of the cluster-specific values $\theta^*_h = (a^*_h, b^*_h)$ is
\[
    f(a^*_h, b^*_h \mid \text{rest}) \propto \left(\frac{\Gamma(a^* + b^*)}{\Gamma(a^*) \Gamma(b^*)}\right)^{n_h} \prod_{i: c_i = h} y_i^{a^*_h - 1} (1 - y_i)^{b^*_h - 1} a^*_h b^*_h e^{- 2(a^*_h + b^*_h)}.
\]
To sample from the distribution above,  we employ a Metropolis-adjusted Langevin algorithm step with step-size $0.05$ on the unconstrained parameters $(\log a^*_h, \log b^*_h) \in \R^2$. 

As for the initialization, when using the Laplace privacy mechanism we simply initialize the $Y_i$'s at random on $[0, 1]$. Instead, in the case of the wavelet-based mechanism, we take advantage of the multi-resolution structure of the Haar wavelet basis. Indeed observe that $\varphi_{j, k}(x) > 0$ if and only if $x \in [2^{-j} k, 2^{-j} (k + 1/2)]$, and $\varphi_{j, k}(x) < 0$ if and only if $x \in [2^{-j} (k + 1/2), 2^{-j} (k + 1)]$. Therefore, we consider the last $2^{j_1}$ elements (i.e., the ones associated to the partition of the domain at the highest resolution) in $Z_i$ and take $k^*$ as the index associated to the maximum absolute value of such elements. Then, we initialize the $Y_i$'s by sampling from a Uniform distribution on $[2^{-j} k, 2^{-j} (k + 1/2)]$ or on  $[2^{-j} (k + 1/2), 2^{-j} (k + 1)]$ depending on whether the corresponding $Z_i$'s were positive or negative.

\section{Numerical Illustrations for the Global Privacy Mechanism}\label{app:global_simu}

We consider the privacy mechanism in \eqref{eq:smooth_hist}. We generate $n=250$ data points from a mixture of three beta distributions as in Section \ref{sec:ex_wav}. 
We consider privacy levels $\varepsilon = 2, 10, 50, 100, 250$ and set the parameters $m, k$ and $\delta$ following Theorem 4.3 in \cite{wasserman2010statistical}, specifically $m = L \lfloor n^{1/5} + 1\rfloor$ for three choices of $L=2, 4, 6$, $k = \lfloor n^{3/5} + 1 \rfloor$ and find $\delta$ numerically as the smallest number in $(0, 1)$ satisfying 
\[
  k \log \left(\frac{1 - \delta}{\delta} \frac{m}{n} - 1 \right) \leq \varepsilon.
\]
The base measure $G_0$ is fixed as in Section \ref{sec:ex_wav} as $G_0(\dd a, \dd b) = \mbox{Gamma}(\dd a \mid 2, 2) \mbox{Gamma}(\dd b \mid 2, 2)$. 

We sample from the posterior distribution using the two algorithms described in \Cref{sec:global_algo}.
Although not reported here, the mixing of the MCMC chain shows the same behavior as in the numerical illustration for the locally differentially private mechanism. Namely, the acceptance rates decrease with the level of privacy $\varepsilon$ for both algorithms, while being consistently  above the theoretical lower bound of $\exp(- \varepsilon)$. The effective sample sizes for the number of clusters is a decreasing function of $\varepsilon$ for the marginal algorithm, while it is essentially constant for the conditional one.
\Cref{fig:global_dens_est} shows the density estimates, which are mostly in accordance for both algorithms.
Similarly to the case of the wavelet-based mechanism, it is evident here that the parameters involved in the privacy mechanism must be carefully chosen, and it is not enough to simply satisfy their asymptotic growth rates. Indeed for $L=2$ and $L=6$ the density estimates are poor, with most of the mass being pushed to the borders and the remaining being spread uniformly through the domain.
When $L=2$ this is likely due to the number of bins $m$ being too small, so that the resulting histogram looks uniform, when $L=6$ the number of bins is significantly larger but $\delta$ is large as well which makes the resulting histogram rather uniform.
For $L=4$ the densities are rather satisfactory, but not as good as the ones with the locally differentially private mechanisms in \Cref{sec:ex_wav}. In all cases, the posterior assigns a sigificant mass to the boundaries of the domain. 
This behavior can be partially contrasted with the choice of a more concentrated base measure. See Figure \ref{fig:global_dens2} where we report the density estimates when  $G_0(\dd a, \dd b) = \mbox{Gamma}(\dd a \mid 90, 2) \mbox{Gamma}(\dd b \mid 90, 2)$.

\begin{figure}[t]
    \centering
    \includegraphics[width=\linewidth]{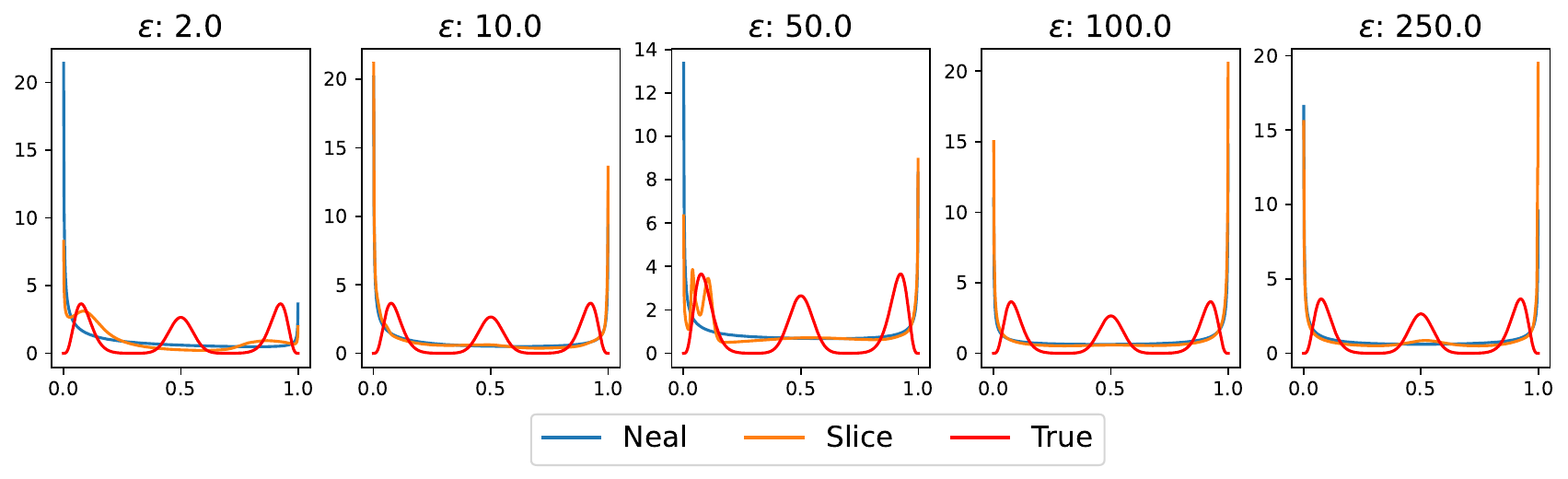}
    \includegraphics[width=\linewidth]{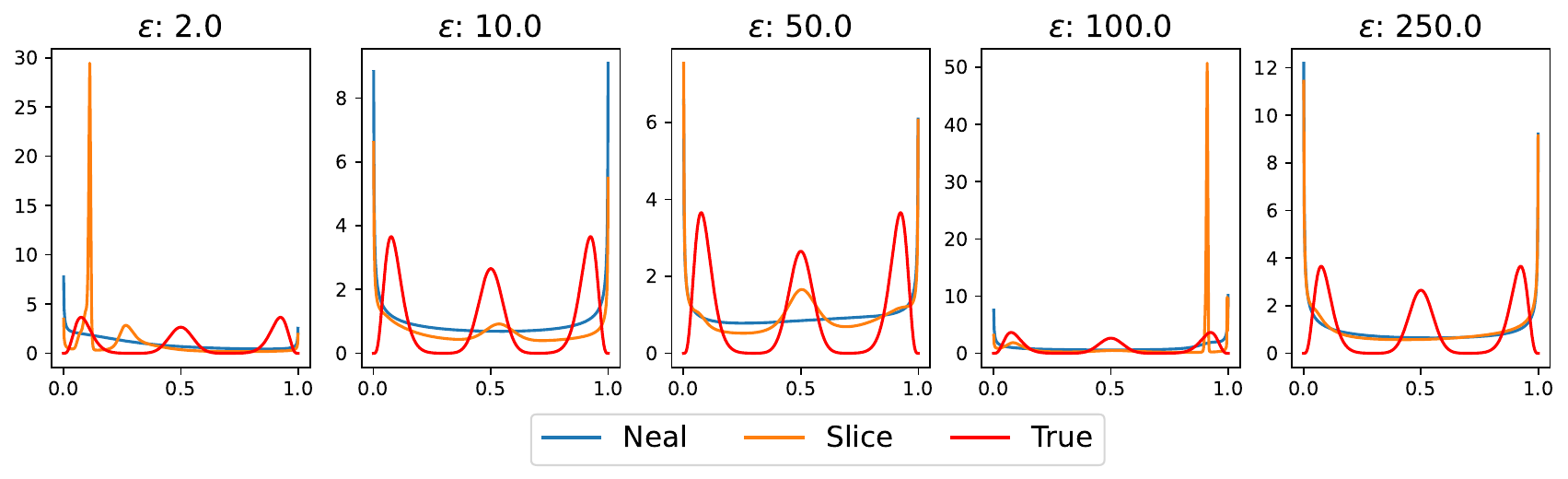}
    \includegraphics[width=\linewidth]{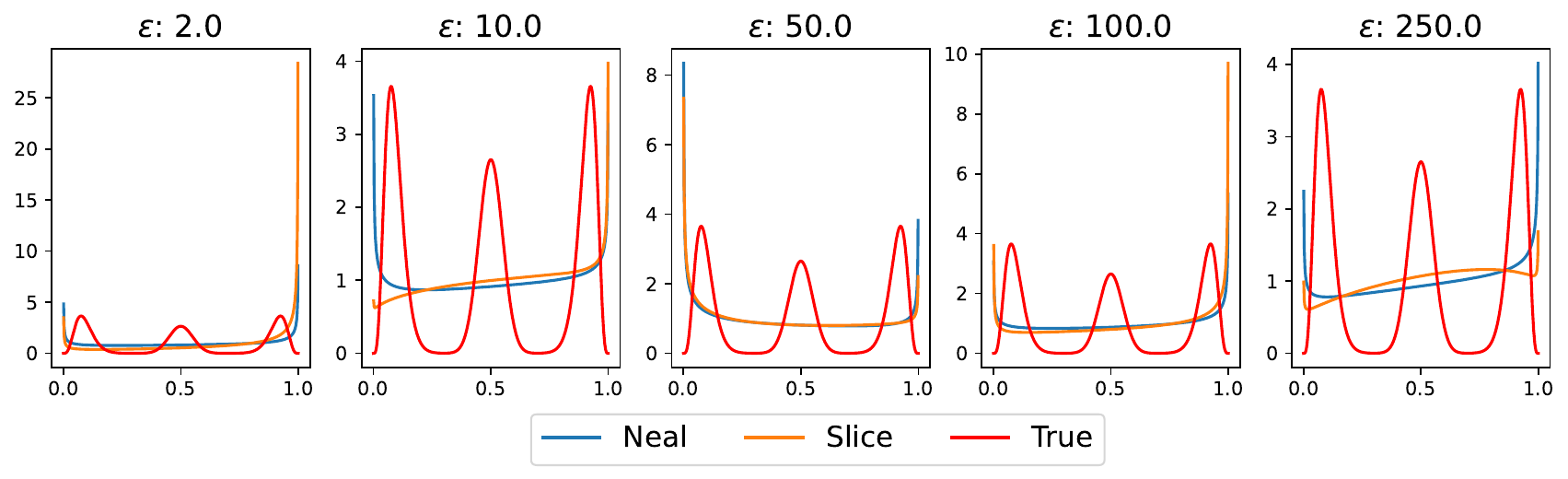}
    \caption{Density estimate for the simulated example in Section \ref{sec:global}. From top to bottom the parameter $L$ in $L=2, 4, 6$}
    \label{fig:global_dens_est}
\end{figure}

\begin{figure}[t]
    \centering
    \includegraphics[width=\linewidth]{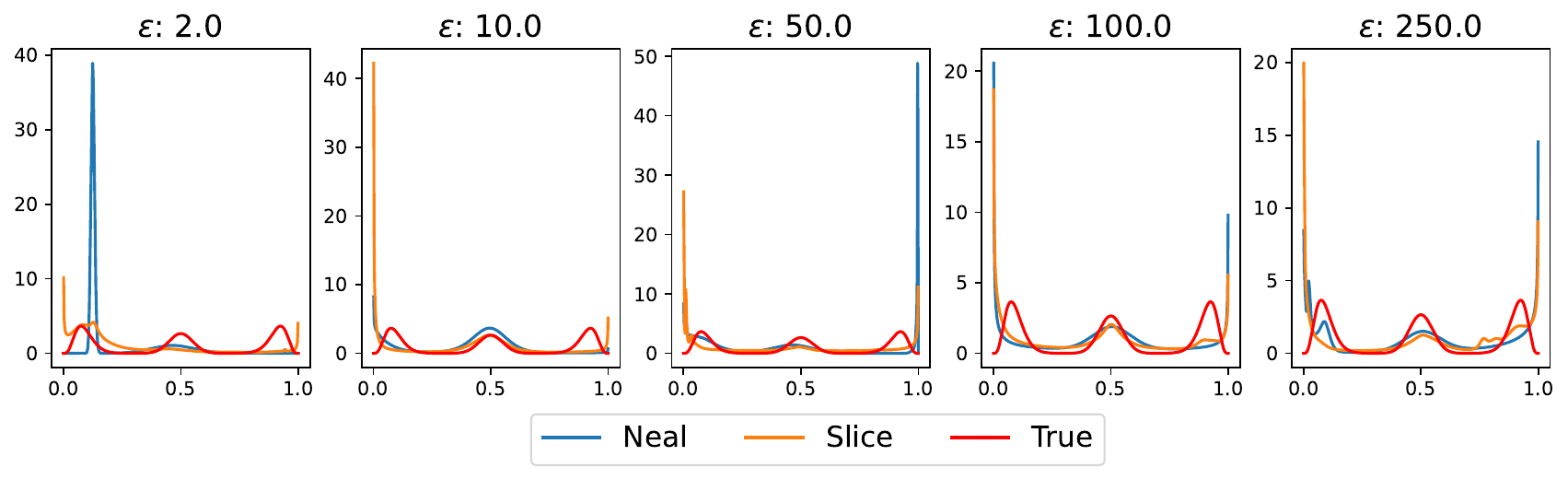}
    \includegraphics[width=\linewidth]{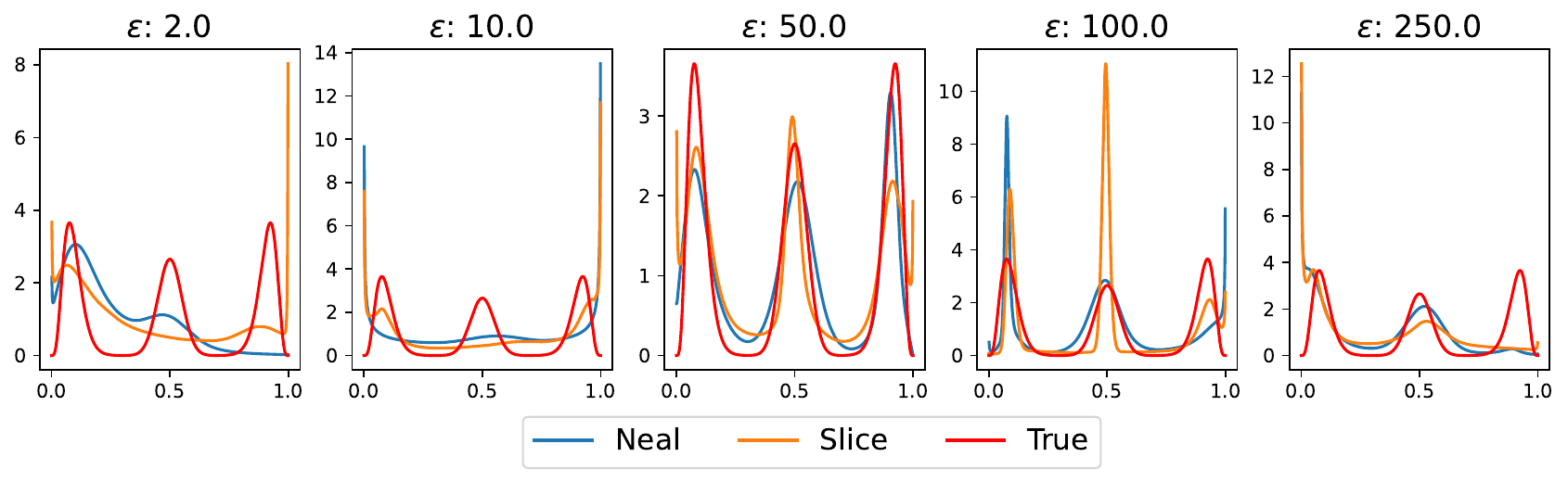}
    \includegraphics[width=\linewidth]{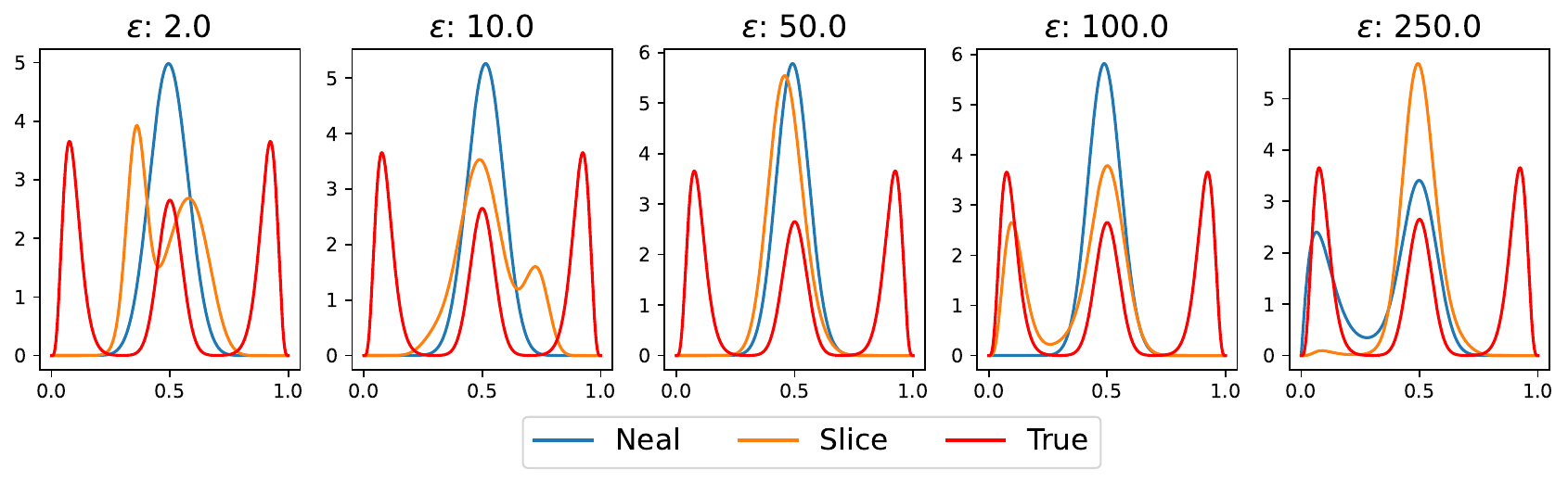}
    \caption{Density estimate for the simulated example in Section \ref{sec:global} when $G_0(\dd a, \dd b) = \mbox{Gamma}(\dd a \mid 90, 2) \mbox{Gamma}(\dd b \mid 90, 2)$. From top to bottom the parameter $L$ in $L=2, 4, 6$}
    \label{fig:global_dens2}
\end{figure}

\FloatBarrier

\section{Futher Plots}

Figure \ref{fig:laplace_dens_estimate_full} shows the density estimates for one particular instance of the simulated datasets in Section \ref{sec:ex_lap}.

Figure \ref{fig:lap_ari} shows the adjusted rand index between the true and estimated clustering for the simulation in Section \ref{sec:ex_lap}

Figure \ref{fig:lap_clus_hist} shows the posterior distribution for the number of clusters for the simulated example in Section \ref{sec:ex_lap}.

Figure \ref{fig:wav_arate} shows the acceptance rate for the simulated example in Section \ref{sec:ex_wav} 

Figure \ref{fig:global_dens2} shows the density estimates for the simulated example in Section \ref{sec:global} when setting the base measure $G_0(\dd a, \dd b) = \mbox{Gamma}(\dd a \mid 90, 2) \mbox{Gamma}(\dd b \mid 90, 2)$.

Figure \ref{fig:gauss_dens} shows the density estimates for one particular instance of the simulated datasets in Section \ref{sec:ex_gauss}.

\begin{figure}
    \centering
    \includegraphics[width=\linewidth]{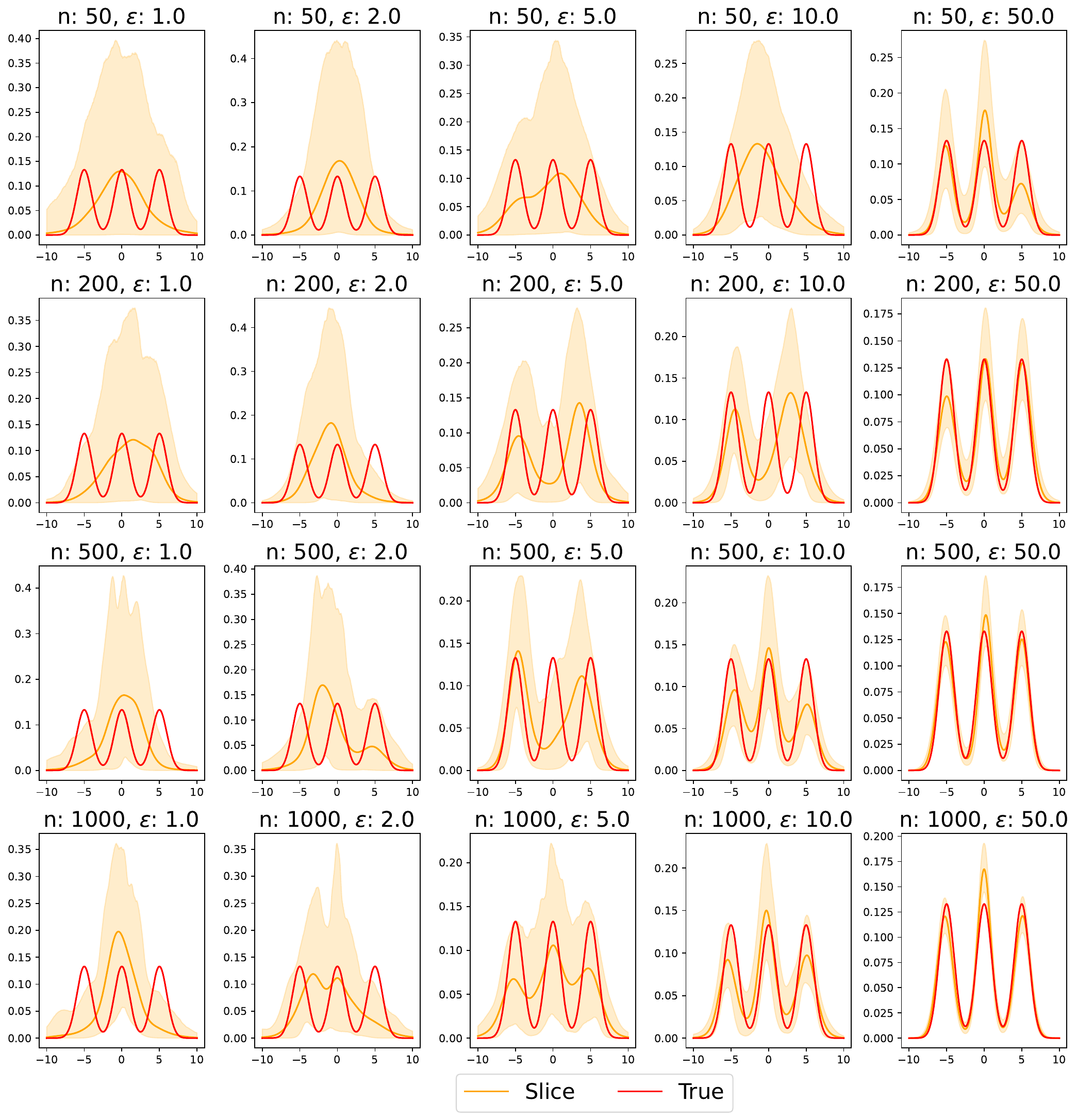}
    \caption{Density estimate for different values of the sample size $n$ and privacy levels $\varepsilon$ for the experiment in Section \ref{sec:ex_lap}. Red line is the true data-generating density, orange line is the posterior mean density obtained with the slice sampling algorithm, shaded orange region is the 95\% pointwise credible set.}
    \label{fig:laplace_dens_estimate_full}
\end{figure}

\begin{figure}
    \centering
    \includegraphics[width=\linewidth]{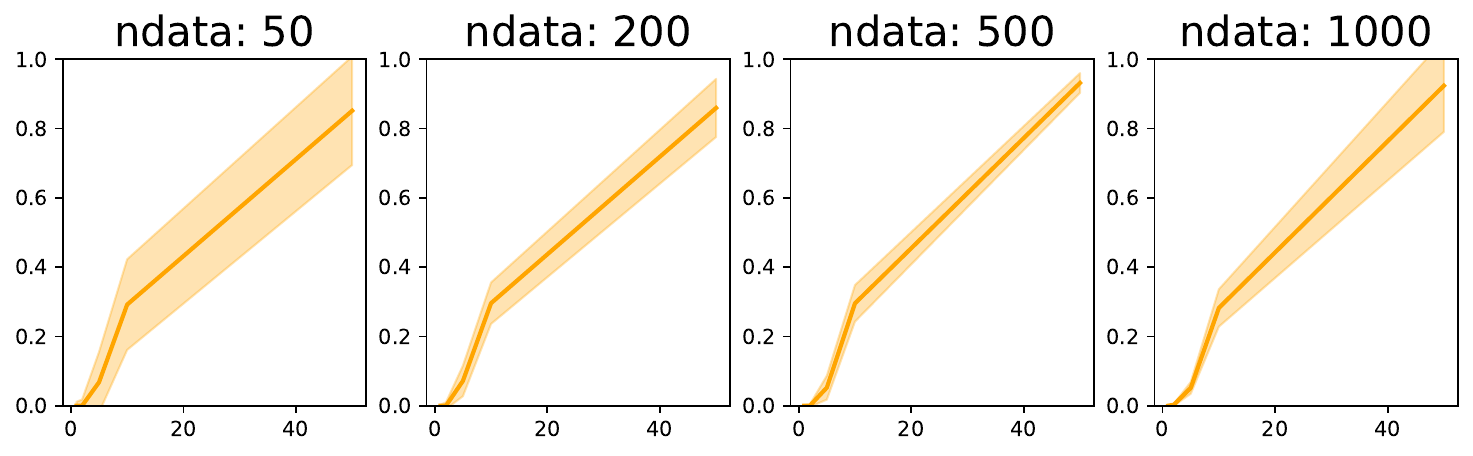}
    \caption{Adjusted Rand Index between estimated and true clustering for the experiment in Section \ref{sec:ex_lap}. Solid orange line is the median over 50 independent experiments, shaded region is a 95\% pointiwse confidence interval.}
    \label{fig:lap_ari}
\end{figure}

\begin{figure}
    \centering
    \includegraphics[width=\linewidth]{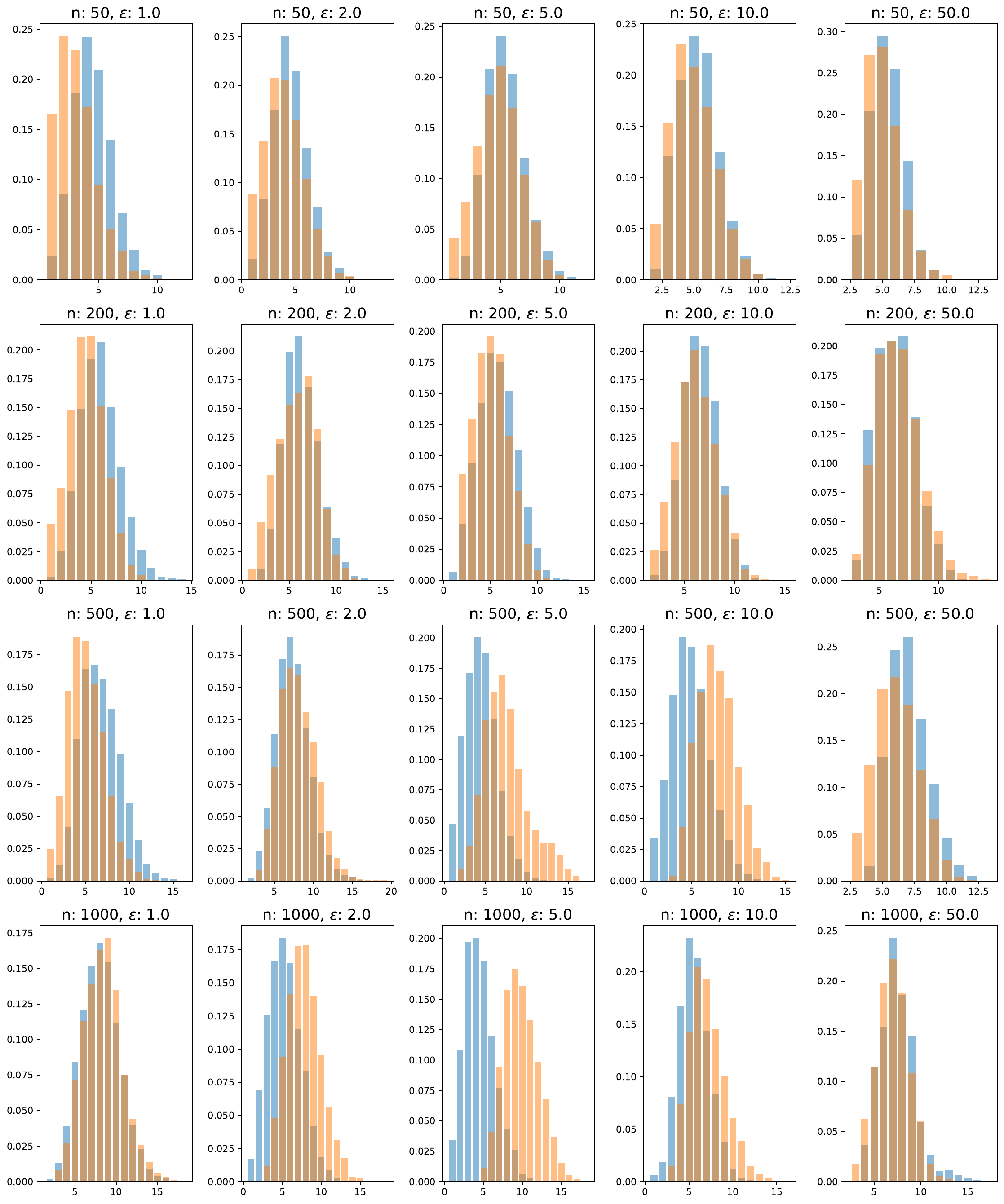}
    \caption{Posterior distribution of the number of clusters for the example in Section \ref{sec:ex_lap}. The blue and orange bars correspond to Neal's and the Slice algorithms, respectively. }
    \label{fig:lap_clus_hist}
\end{figure}

\begin{figure}
    \centering
    \includegraphics[width=0.6\linewidth]{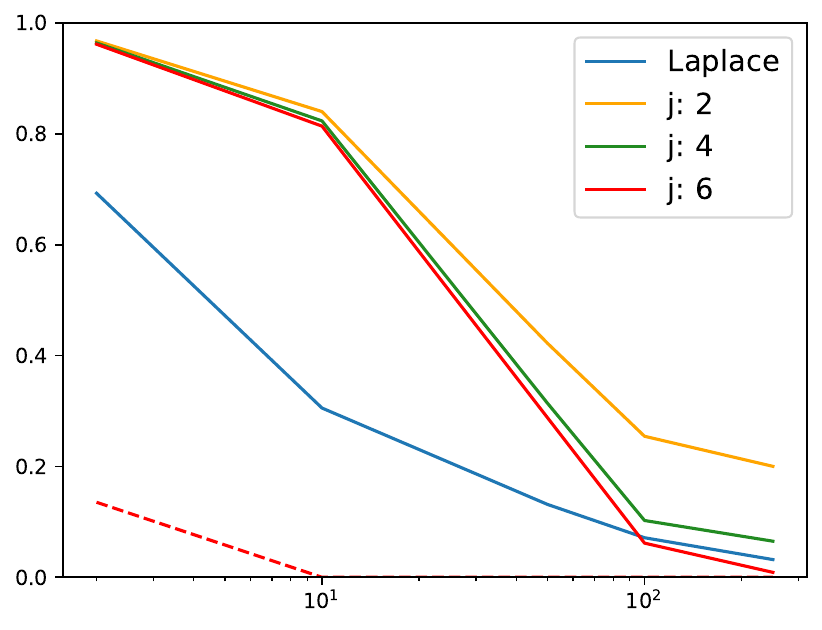}
    \caption{Acceptance rates for the simulated example in Section \ref{sec:ex_wav}. The dashed line corresponds to the theoretical lower bound $\exp(-\varepsilon).$}
    \label{fig:wav_arate}
\end{figure}

\begin{figure}[t]
    \centering
    \includegraphics[width=\linewidth]{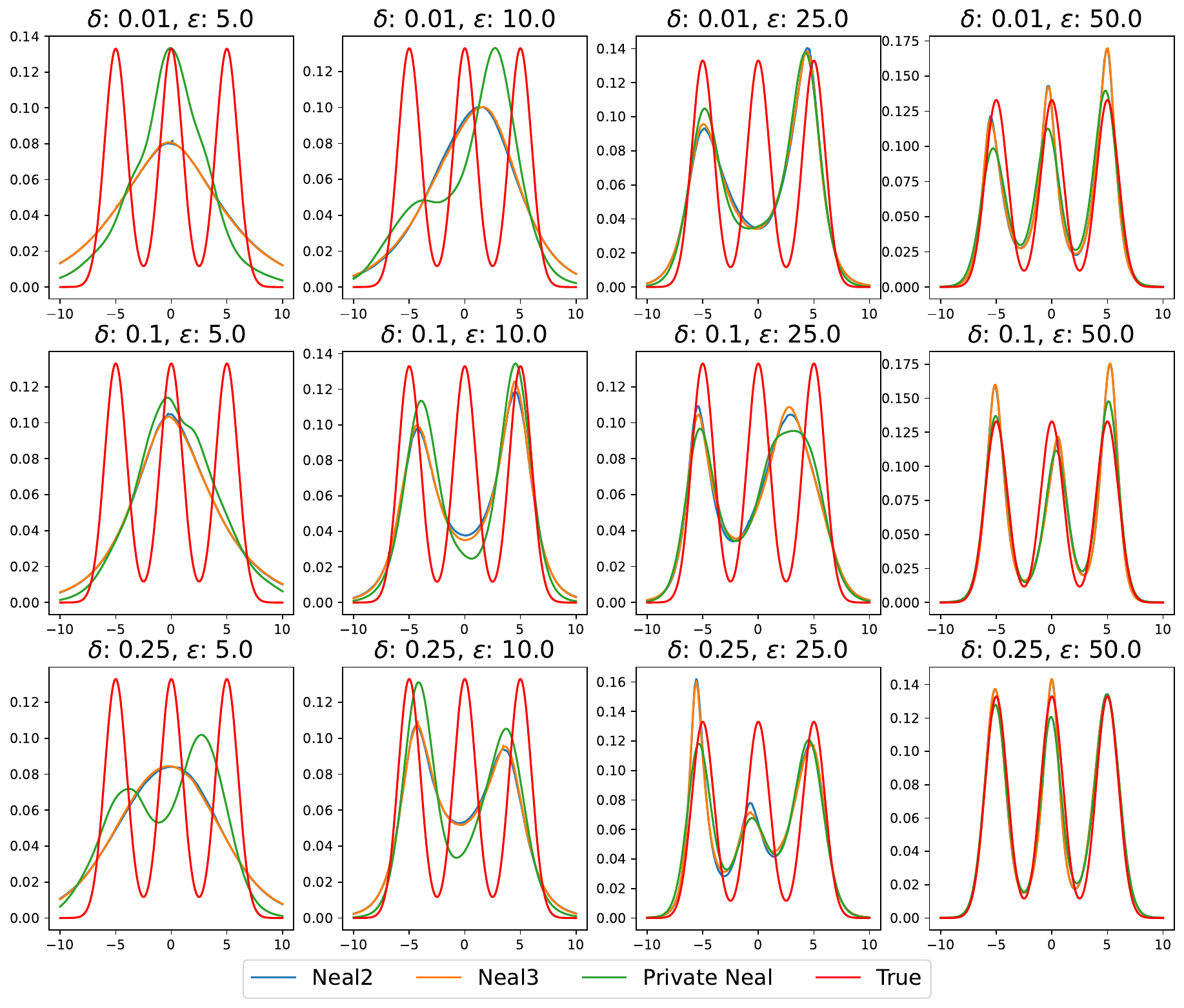}
    \caption{Density estimate for the simulated example in Section \ref{sec:ex_gauss} for different values of $\varepsilon$ and $\delta$.}
    \label{fig:gauss_dens}
\end{figure}

\begin{figure}[t]
    \centering
     \includegraphics[width=0.6\linewidth]{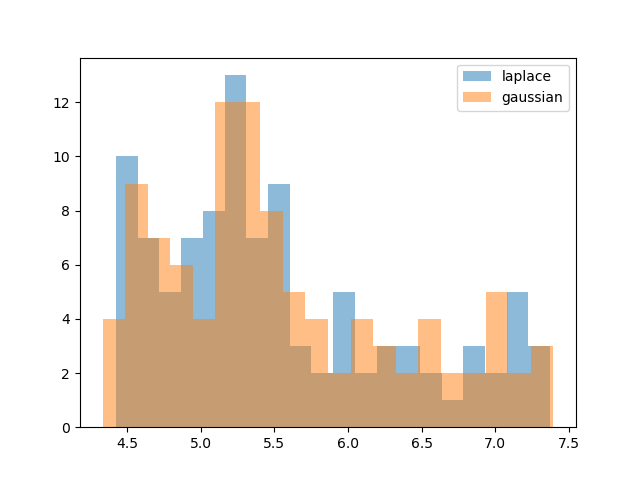}
    \caption{Histogram of the AVIS dataset in Section \ref{sec:blood_donors} after the perturbation with Laplace and Gaussian noise}
    \label{fig:blood-pert}
\end{figure}

\end{document}